\documentclass[a4paper,11pt]{article}
\pdfoutput=1

\usepackage{amsthm}
\usepackage{amsmath}
\usepackage{amssymb}
\usepackage{amsfonts}
\usepackage{amstext}
\usepackage{xspace}
\usepackage{graphicx}

\usepackage{fullpage}

\RequirePackage{tikz}

\RequirePackage[breaklinks,pdfborder={0 0 0}]{hyperref}

\usepackage[ backend=biber, style=ieee-alphabetic, firstinits=false ]{biblatex}
\bibliography{seth}

\usepackage{sectsty}
\allsectionsfont{\bfseries\sffamily}

\usepackage{enumitem}

\newtheorem{theorem}{Theorem}[section]
\newtheorem{lemma}[theorem]{Lemma}
\newtheorem{corollary}[theorem]{Corollary}
\newtheorem{definition}[theorem]{Definition}
\newtheorem{observation}[theorem]{Observation}

\newcommand{\pp}[1]{{\normalfont\textsc{#1}}}

\newcommand{\cnfsat}{\pp{CNF-Sat}\xspace}
\newcommand{\cnfsatargs}[1]{\pp{#1-CNF-Sat}\xspace}

\newcommand{\sparsecnfsat}{\pp{Sparse-CNF-Sat}\xspace}
\newcommand{\sparsecnfsatargs}[2]{\pp{#2-Sparse-#1-CNF-Sat}\xspace}

\newcommand{\pcnfsat}{\pp{CNF-$\oplus$Sat}\xspace}
\newcommand{\pcnfsatargs}[1]{\pp{#1-CNF-$\oplus$Sat}\xspace}

\newcommand{\psparsecnfsatargs}[2]{\pp{#2-Sparse-#1-CNF-$\oplus$Sat}\xspace}

\newcommand{\naesat}{\pp{NAE-Sat}\xspace}
\newcommand{\naesatargs}[1]{\pp{#1-NAE-Sat}\xspace}

\newcommand{\pnaesat}{\pp{$\oplus$NAE-Sat}\xspace}
\newcommand{\pnaesatargs}[1]{\pp{#1-$\oplus$NAE-Sat}\xspace}

\newcommand{\hittingset}{\pp{Hitting Set}\xspace}
\newcommand{\hittingsetargs}[1]{\pp{#1-Hitting Set}\xspace}

\newcommand{\sparsehittingset}{\pp{Sparse-Hitting Set}\xspace}
\newcommand{\sparsehittingsetargs}[2]{\pp{#2-Sparse-#1-Hitting 
    Set}\xspace}

\newcommand{\phittingset}{\pp{$\oplus$Hitting Sets}\xspace}
\newcommand{\phittingsetargs}[1]{\pp{#1-$\oplus$Hitting Sets}\xspace}
\newcommand{\psparsehittingsetargs}[2]{\pp{#2-Sparse-#1-$\oplus$Hitting 
    Sets}\xspace}

\newcommand{\pallhittingset}{{\text{$\oplus$\sc All Hitting 
      Sets}}\xspace}
\newcommand{\pallsparsehittingsetargs}[2]{\pp{#2-Sparse-#1-$\oplus$All 
    Hitting Sets}\xspace}

\newcommand{\cvc}{\pp{Connected Vertex Cover}\xspace}
\newcommand{\pcvc}{\pp{$\oplus$Connected Vertex Covers}\xspace}

\newcommand{\setcover}{\pp{Set Cover}\xspace}
\newcommand{\setcoverarg}[1]{\pp{#1-Set Cover}\xspace}

\newcommand{\psetcover}{\pp{$\oplus$Set Covers}\xspace}
\newcommand{\psetcoverarg}[1]{\pp{#1-$\oplus$Set Covers}\xspace}
\newcommand{\psparsesetcoverargs}[2]{\pp{#2-Sparse-#1-$\oplus$Set 
    Covers}\xspace}

\newcommand{\pallsetcover}{\pp{$\oplus$All Set Covers}\xspace}
\newcommand{\pallsetcoverarg}[1]{\pp{#1-$\oplus$All Set 
    Covers}\xspace}

\newcommand{\pallsparsesetcoverargs}[2]{\pp{#2-Sparse-#1-$\oplus$All 
    Set Covers}\xspace}

\newcommand{\setpart}{\pp{Set Partitioning}\xspace}
\newcommand{\setpartarg}[1]{\pp{#1-Set Partitioning}\xspace}

\newcommand{\setsplitting}{\pp{Set Splitting}\xspace}
\newcommand{\setsplittingarg}[1]{\pp{#1-Set Splitting}\xspace}

\newcommand{\psetsplitting}{\pp{$\oplus$Set Splitting}\xspace}
\newcommand{\psetsplittingarg}[1]{\pp{#1-$\oplus$Set 
    Splitting}\xspace}

\newcommand{\steinertree}{\pp{Steiner Tree}\xspace}
\newcommand{\psteinertree}{\pp{$\oplus$Steiner Tree}\xspace}

\newcommand{\subsetsum}{\pp{Subset Sum}\xspace}

\newcommand{\vspcircuitsat}{\pp{VSP-Circuit-SAT}\xspace}
\newcommand{\vspcircuitsatargs}[1]{\pp{#1-VSP-Circuit-SAT}\xspace}

\newcommand{\hamcycle}{\pp{Hamiltonian Cycle}\xspace}

\newcommand{\coloring}{\pp{Graph Coloring}\xspace}

\newcommand{\seth}{\textrm{SETH}\xspace}
\newcommand{\pseth}{\textrm{\ensuremath{\oplus}-SETH}\xspace}
\newcommand{\vsp}{{\normalfont\textsc{VSP}}\xspace}
\newcommand{\inp}{{\text{\rm input}}\xspace}
\newcommand{\out}{{\text{\rm output}}\xspace}

\newcommand{\nats}{\mathbb{N}}

\newcommand{\OH}{O}

\newcommand{\C}{\mathcal{C}}

\newcommand{\poly}{{\rm{poly}}}
\newcommand{\ones}{{\rm{ones}}}
\newcommand{\cF}{{\mathcal{F}}}
\newcommand{\cS}{{\mathcal{S}}}

\RequirePackage{footmisc}

\title{%
 \bfseries\sffamily%
  On Problems as Hard as CNF-SAT%
  \thanks{An extended abstract of this paper appears in the
  proceedings of CCC 2012.}%
}

\author{%
  {\large\sc Marek Cygan%
    \thanks{
      IDSIA, University of Lugano, Switzerland.
      \texttt{marek@idsia.ch}.
      Partially supported by National Science Centre grant no. N206
      567140, Foundation for Polish Science and ONR Young Investigator
      award when at the University at Maryland.
    }}
  \and {\large\sc Holger Dell%
    \thanks{%
      LIAFA, Universit\'e Paris Diderot, France.
      \texttt{holger@liafa.univ-paris-diderot.fr}.
      Research partially supported by the Alexander von Humboldt
      Foundation and NSF grant 1017597.
    }}
  \and {\large\sc Daniel Lokshtanov%
    \thanks{%
      University of Bergen, Norway.
      \texttt{daniello@ii.uib.no}.
    }}
  \and {\large\sc D\'aniel Marx%
    \thanks{%
      Computer and Automation Research Institute, Hungarian
      Academy of Sciences (MTA SZTAKI), Budapest, Hungary.
      \texttt{dmarx@cs.bme.hu}.
      Research supported by ERC Starting Grant PARAMTIGHT (280152).
    }}
  \and {\large\sc Jesper Nederlof%
    \thanks{%
      Utrecht University, The Netherlands.
      \texttt{j.nederlof@uu.nl}.
      Supported by NWO project "Space and Time Efficient Structural
      Improvements of Dynamic Programming Algorithms".
    }}
  \and {\large\sc Yoshio Okamoto%
    \thanks{%
      Japan Advanced Institute of Science and Technology, Japan.
      \texttt{okamotoy@uec.ac.jp}.
      Partially supported by Grant-in-Aid for Scientific Research from
      Japan Society for the Promotion of Science.
    }}
  \and {\large\sc Ramamohan Paturi%
    \thanks{%
      University of California, USA.
      \texttt{paturi@cs.ucsd.edu}.
      This research is supported by NSF grant CCF-0947262 from the
      Division  of Computing and Communication Foundations.  Any
      opinions, findings and conclusions or recommendations expressed
      in this material are those of the authors and do not necessarily
      reflect the views of the National Science Foundation.
    }}
  \and {\large\sc Saket Saurabh%
    \thanks{%
      Institute of Mathematical Sciences, India.
      \texttt{saket@imsc.res.in}.
    }}
  \and {\large\sc Magnus Wahlstr\"om%
    \thanks{%
      Royal Holloway, University of London.
      \texttt{Magnus.Wahlstrom@rhul.ac.uk}.
    }}
}

\begin{document}
\maketitle

\setcounter{page}{0}
\thispagestyle{empty}
\begin{abstract}
  \centering
  \begin{minipage}{.95\textwidth}
    \setlength{\parindent}{1.5em}
The field of exact exponential time algorithms for NP-hard problems
has thrived over the last decade.
While exhaustive search remains asymptotically the fastest known
algorithm for some basic problems, difficult and non-trivial
exponential time algorithms have been found for a myriad of problems,
including \textsc{Graph Coloring}, \textsc{Hamiltonian Path},
\textsc{Dominating Set} and $3$-$\cnfsat{}$.
In some instances, improving these algorithms further seems to be out
of reach.
The $\cnfsat{}$ problem is the canonical example of a problem for
which the trivial exhaustive search algorithm runs in time $O(2^n)$,
where $n$ is the number of variables in the input formula.
While there exist non-trivial algorithms for $\cnfsat{}$ that run in
time $o(2^n)$, no algorithm was able to improve the \emph{growth
rate}~$2$ to a smaller constant, and hence it is natural to conjecture
that $2$ is the optimal growth rate.
The \emph{strong exponential time hypothesis}~(SETH) by Impagliazzo and
Paturi [JCSS 2001] goes a little bit further and asserts that, for
every $\epsilon<1$, there is a (large) integer $k$ such that
$k$-$\cnfsat{}$ cannot be computed in time $2^{\epsilon n}$.

In this paper, we show that, for every $\epsilon < 1$, the problems
$\hittingset$, $\setsplitting$, and $\naesat$ cannot be computed in
time $O(2^{\epsilon n})$ unless SETH fails.
Here $n$ is the number of elements or variables in the input.
For these problems, we actually get an equivalence to SETH in a
certain sense.
We conjecture that SETH implies a similar statement for $\setcover$,
and prove that, under this assumption, the fastest known algorithms
for $\steinertree$, $\cvc$, $\setpart$, and the pseudo-polynomial time
algorithm for $\subsetsum$ cannot be significantly improved.
Finally, we justify our assumption about the hardness of $\setcover$
by showing that the parity of the number of solutions to $\setcover$
cannot be computed in time $O(2^{\epsilon n})$ for any $\epsilon<1$
unless SETH fails.

  \end{minipage}
\end{abstract}

\newpage

\section{Introduction}

Every problem in NP can be solved in time $2^{\poly(m)}$ by brute
force, that is, by enumerating all candidates for an NP-witness, which
is guaranteed to have length polynomial in the input size~$m$.
While we do not believe that polynomial time algorithms for
NP-complete problems exist, many NP-complete problems have exponential
time algorithms that are dramatically faster than the na\"{i}ve brute
force algorithm.
For some classical problems, such as \subsetsum{} or
\hamcycle{}, such algorithms were
known~\cite{HeldKarp62jsiam,Bellman62} even before the
concept of NP-completeness was discovered.
Over the last decade, a subfield of algorithms devoted to developing
faster exponential time algorithms for NP-hard problems has emerged.
A myriad of problems have been shown to be solvable much faster than
by na\"{i}ve brute force, and a variety of algorithm design techniques
for exponential time algorithms has been developed.

What the field of exponential time algorithms sorely lacks is a
complexity-theoretic framework for showing running time lower bounds.
Some problems, such as {\sc Independent Set} and {\sc Dominating Set}
have seen a chain of
improvements~\cite{FominGK09,RooijND09,Robson86,KneisLR09}, each new
improvement being smaller than the previous.
For these problems, the running time of the discovered algorithms
seems to converge towards $O(C^n)$ for some unknown constant~$C$,
where~$n$ denotes the number of vertices of the input graphs.
For other problems, such as \coloring{} or \steinertree{}, non-trivial
algorithms have been found, but improving the \emph{growth rate~$C$}
of the running time any further seems to be out of
reach~\cite{BjorklundHK09,Nederlof09}.
The purpose of this paper is to develop tools that allow us to explain
why we are stuck for these problems.
Ideally, for any problem whose best known algorithm runs in
time~$O(C^n)$, we want to prove that the existence of $O(c^n)$-time
algorithms for any constant~$c<C$ would have implausible
complexity-theoretic consequences.

\paragraph{Previous Work.}
Impagliazzo and Paturi's {\em Exponential Time Hypothesis} (ETH)
addresses the question whether NP-hard problems can have algorithms
that run in ``subexponential time'' \cite{IP01}.
More precisely, the hypothesis asserts that $3$-\cnfsat cannot be
computed in time~$2^{o(n)}$, where $n$ is the number of variables
in the input formula.
ETH is considered to be a plausible complexity-theoretic assumption,
and subexponential time algorithms have been ruled out under ETH for
many decision problems~\cite{IPZ01},
parameterized problems~\cite{ChenCFHJKX05,LokshtanovMS11},
approximation problems~\cite{Marx07a}, and
counting problems~\cite{DHMTW12}.
However, ETH does not seem to be sufficient for pinning down what
exactly the best possible growth rate is.
For this reason, we base our results on a stronger hypothesis.

The fastest known algorithms for \cnfsat{} have running times of the
form $2^{n-o(n)} \poly(m)$
\cite{Schuler:2005:ASP:1056301.1056304,DBLP:conf/coco/Williams11},
which does not improve upon the growth rate~$2$ of the na\"ive brute
force algorithm that runs in time $2^n \poly(m)$.
Hence a natural candidate for a stronger hypothesis is that \cnfsat
cannot be computed in time $2^{\epsilon n}\poly(m)$ for any
$\epsilon<1$.
However, we do not know whether our lower bounds on the growth rate of
specific problems can be based on this hypothesis.
The main technical obstacle is that we have no analogue of the
sparsification lemma, which applies to \mbox{$k$-CNF} formulas and
makes ETH a robust hypothesis~\cite{IPZ01}.
In fact, very recent results indicate that such a sparsification may
be impossible for general CNF formulas~\cite{SS11}.
For this reason, we consider the {\em Strong Exponential Time
Hypothesis} (\seth) of Impagliazzo and
Paturi~\cite{IP01,IPZ01,CalabroIP09}.
This hypothesis asserts that, for every $\epsilon<1$, there is a
(large) integer~$k$ such that $k$-\cnfsat cannot be computed by any
bounded-error randomized algorithm in time~$O(2^{\epsilon n})$.
In particular, SETH implies the hypothesis for \cnfsat above, but we
do not know whether they are equivalent.
Since SETH is a statement about $k$-CNF formulas for
constant~$k=k(\epsilon)$, we can apply the sparsification lemma for
every fixed~$k$, which allows us to use SETH as a starting point in
our reductions.

\paragraph{Our results.}
Our first theorem is that SETH is equivalent to lower bounds on the
time complexity of a number of standard NP-complete problems.
\begin{theorem}\label{thm:branch}
  Each of the following statements is equivalent to SETH:
\begin{enumerate}[%
      label=(\roman*) $\,\forall \epsilon < 1 .\exists k.\,$,%
      leftmargin=8em]
    \item
      $k$-\cnfsat,
      the satisfiability problem for $n$-variable $k$-CNF formulas,
      cannot be computed in time~$O(2^{\epsilon n})$.
    \item
      $k$-\hittingset,
      the hitting set problem for set systems over~$[n]$ with sets of
      size at most~$k$,
      cannot be computed in time~$O(2^{\epsilon n})$.
    \item
      $k$-\setsplitting,
      the set splitting problem for set systems over~$[n]$ with sets
      of size at most~$k$,
      cannot be computed in time~$O(2^{\epsilon n})$.
    \item
      $k$-\naesat,
      the not-all-equal assignment problem for $n$-variable $k$-CNF
      formulas,
      cannot be computed in time~$O(2^{\epsilon n})$.
    \item[(v) $\,\forall \epsilon < 1 .\exists c.\,$]
      \vspcircuitsat{$c$},
      the satisfiability problem for $n$-variable series-parallel
      circuits of size at most~$cn$,
      cannot be computed in time~$O(2^{\epsilon n})$.
  \end{enumerate}
\end{theorem}
For all of the above problems, the na\"ive brute force algorithm runs
in time~$O(2^n)$.
While there may not be a consensus that SETH is a ``plausible''
complexity-theoretic assumption, our theorem does indicate that
finding an algorithm for \cnfsat whose growth rate is smaller than~$2$
is as difficult as finding such an algorithm for any of the above
problems.
Since our results are established via suitable reductions, this can be
seen as a completeness result under these reductions.
Moreover, we actually prove that the optimal growth rates for all of
the problems above are \emph{equal} as~$k$ tends to infinity.
This gives an additional motivation to study the Strong Exponential
Time Hypothesis.

An immediate consequence of Theorem~\ref{thm:branch} is that, if SETH
holds, then \cnfsat, \hittingset, \setsplitting, \naesat, and
the satisfiability problem of series-parallel circuits
do not have bounded-error randomized algorithms that run in time
$2^{\epsilon n}\poly(m)$ for any $\epsilon<1$.
All of these problems are {\em search} problems, where the objective
is to find a particular object in a search space of size~$2^n$.
Of course, we would also like to show tight connections between SETH
and the optimal growth rates of problems that \emph{do} have
non-trivial exact algorithms.
Our prototypical such problem is \setcover:
Given a set system with~$n$ elements and~$m$ sets, we want to select a
given number~$t$ of sets that cover all elements.
Exhaustively trying all possible ways to cover the elements takes time
at most $2^m\poly(m)$.
However, $m$ could be much larger than~$n$, and it is natural to ask
for the best running time that one can achieve in terms of~$n$.
It turns out that a simple dynamic programming
algorithm~\cite{FominKW04} can solve \setcover in time $2^n\poly(m)$.
The natural question is whether the growth rate of this simple
algorithm can be improved.
While we are not able to resolve this question, we connect the
existence of an improved algorithm for \setcover to the existence of
faster algorithms for several problems.
Specifically, we show the following theorem.
\begin{theorem}\label{thm:dpsolve}
  Assume that, for all $\epsilon < 1$, there is a $k$ such that
  \setcover with sets of size at most $k$ cannot be computed in time
  $2^{\epsilon n}\poly(m)$.
  Then, for all $\epsilon<1$, we have:
\begin{enumerate}[label=(\roman*)]\setlength\itemsep{-.7mm}
    \item \steinertree cannot be computed in time $2^{\epsilon t}\poly(n)$,
    \item \cvc cannot be computed in time $2^{\epsilon t}\poly(n)$,
    \item \setpart cannot be computed in time $2^{\epsilon n}\poly(m)$, and
    \item \subsetsum cannot be computed in time $t^\epsilon\poly(n)$.
  \end{enumerate}
\end{theorem}
All problems mentioned in this theorem have non-trivial algorithms
whose running times are as above with $\epsilon=1$
\cite{BjorklundHKK07,Nederlof09,CNP+11,FominKW04,introductionToAlg}.
Under the assumption in the theorem, we therefore obtain tight lower
bounds on the growth rate of exact algorithms for \steinertree, \cvc,
\setpart, and \subsetsum.
The best currently known algorithms for these problems share two
interesting common features.
First, they are all {\em dynamic programming} algorithms.
Thus, Theorem~\ref{thm:dpsolve} hints at \setcover being a
``canonical'' dynamic programming problem.
Second, the algorithms can all be modified to compute the number of
solutions modulo two in the same running time.
In fact, the currently fastest algorithm~\cite{CNP+11} for \cvc works
by reducing the problem to computing the number of solutions modulo
two.

While Theorem~\ref{thm:branch} is an equivalence,
Theorem~\ref{thm:dpsolve} is not.
One might ask whether it is possible to find reductions back to
\setcover and to strengthen Theorem~\ref{thm:dpsolve} in this manner.
We believe that this would be quite difficult:
A suitable reduction from, say, \steinertree to \setcover that proves
the converse of Theorem~\ref{thm:dpsolve} would probably also work for
$\epsilon=1$.
This would give an alternative proof that \steinertree can be computed
in time $2^t\poly(m)$.
Hence, finding such a reduction is likely to be a challenge since the
fastest known algorithms~\cite{BjorklundHKK07,Nederlof09} for
\steinertree are quite non-trivial --- it took more than 30 years
before the classical $3^{t}\poly(n)$-time Dreyfus--Wagner algorithm
for \steinertree was improved to $2^{t}\poly(n)$.
Similar comments apply to~\cvc since its $2^{t}\poly(n)$ time
algorithm is quite complex~\cite{CNP+11}.

The hardness assumption for \setcover in Theorem~\ref{thm:dpsolve}
needs some justification.
Ideally we would like to replace this assumption with SETH, that is,
we would like to prove that SETH implies the hardness assumption for
\setcover in Theorem~\ref{thm:dpsolve}.
We do not know a suitable reduction, but we are able to provide a
different kind of evidence for hardness:
We show that a $2^{\epsilon n} \poly(m)$-time algorithm to compute the
number of set covers modulo two would violate $\pseth$, which is a
hypothesis that implies $\seth$.
Formally, $\pseth$ asserts that, for all ${\epsilon<1}$, there exists a
(large) integer $k$ such that $k$-\pcnfsat cannot be computed in
time~$O(2^{\epsilon n})$.
Here, $k$-\pcnfsat is the problem of computing the number of
satisfying assignments of a given $k$-CNF formula modulo two.
It follows from known results~\cite{Calabro_isolation,Traxler} (see
also Section~\ref{sec:satprelim}) that, if $\seth$ holds, then so does
$\pseth$.
As a partial justification for the hardness assumption for \setcover in
Theorem~\ref{thm:dpsolve}, we provide the following theorem.
\begin{theorem}\label{thm:parityequiv}
  Each of the following statements is equivalent to $\pseth$:
  \begin{enumerate}[%
      label=(\roman*) $\,\forall \epsilon < 1 .\exists k.\,$,%
      leftmargin=8em]
    \item
      $k$-\pcnfsat,
      the parity satisfiability problem for $n$-variable $k$-CNF
      formulas,
      cannot be computed in time~$O(2^{\epsilon n})$.
    \item
      $k$-\pallhittingset,
      the parity hitting set problem for set systems over~$[n]$ with
      sets of size at most~$k$,
      cannot be computed in time~$O(2^{\epsilon n})$.
    \item
      $k$-\pallsetcover,
      the parity set cover problem for set systems over~$[n]$ with
      sets of size at most~$k$,
      cannot be computed in time~$O(2^{\epsilon n})$.
  \end{enumerate}
\end{theorem}
In the statement of Theorem~\ref{thm:parityequiv}, the \pallhittingset
and \pallsetcover problems are defined as follows: the input is a set
system and the objective is to compute the parity of the number of
hitting sets (resp.\ set covers) in the system.
An immediate consequence of Theorem~\ref{thm:parityequiv} that we find
interesting is that $\pseth$ rules out the existence of
$2^{\epsilon n}\poly(m)$-time algorithms to compute the number of set
covers of a set system, for any $\epsilon<1$.

Theorem~\ref{thm:parityequiv} together with the fact that the
algorithms for all problems mentioned in Theorem~\ref{thm:dpsolve} can
be modified to count solutions modulo two leads to the following
questions:
Can we show running time lower bounds for the counting versions of
these problems?
We show that this is indeed possible. In particular we show that,
assuming $\pseth$, there is no $2^{\epsilon t}\poly(n)$-time algorithm
that computes the parity of the number of Steiner trees that have size
at most~$t$,
and no $2^{\epsilon t}\poly(n)$-time algorithm that computes the
parity of the number of connected vertex covers that have size at
most~$t$.
Thus, unless $\pseth$ fails, any improved algorithm for \setcover,
\steinertree, or \cvc cannot be used to compute the parity of the
number of solutions.

We find it intriguing that $\seth$ and $\pseth$ can be used to show
tight running time lower bounds, sometimes for problems for which the
best algorithm has been improved several times, such as for
\steinertree or \cvc.
We feel that such sharp bounds are unlikely to just be a coincidence,
leading us to conjecture that the relationship between the considered
problems is even closer than what we show.
Specifically, we conjecture that $\seth$ implies the hardness
assumption for \setcover in Theorem~\ref{thm:dpsolve}.
This conjecture provides an interesting open problem.

Our results are obtained by a collection of reductions.
Section~\ref{sec:improvbranch} contains the reductions that constitute
the proof of Theorem~\ref{thm:branch}, and some of the reductions
needed for Theorem~\ref{thm:parityequiv}.
Section~\ref{sec:dp} contains the proof of Theorem~\ref{thm:dpsolve},
the remaining reductions for Theorem~\ref{thm:parityequiv}, and the
hardness results for counting Steiner trees and connected vertex
covers.
A schematic representation of our reductions can be found in
Figure~\ref{fig:redgraph}.

\begin{figure}[tbp]
  \label{fig:redgraph}
  \begin{center}

\begin{tikzpicture}[thick,outer sep=2pt,inner sep=2pt]

 \newcommand{\obs}{O } 
 \newcommand{\thm}{T } 
 \newcommand{\cor}{C } 

  \node[draw] (cnfsat)  at (-2,-1.5) {\cnfsat};
  \node (hittingset)    at (-4,-3.25) {\hittingset};
  \node (setsplitting)  at (-6,-1.5) {\setsplitting};
  \node (naesat)        at (-4, 0.25) {\naesat};

  \node (vspcircuitsat) at (-1,-4.5) {\vspcircuitsat};

  \draw[<-] (cnfsat)
  to[bend left]
  node[sloped,midway,above]
  {\thm\ref{thm:sattohit}}
  (hittingset);
  \draw[<-] (hittingset)
  to[bend left]
  node[sloped,midway,above]
  {\thm\ref{thm:hit-to-setsplitting}}
  (setsplitting);
  \draw[<-] (setsplitting)
  to[bend left]
  node[sloped,midway,below]
  {\obs\ref{obs:setsplitting-to-naesat-to-cnfsat}}
  (naesat);
  \draw[<-] (naesat)
  to[bend left]
  node[sloped,midway,below]
  {\obs\ref{obs:setsplitting-to-naesat-to-cnfsat}}
  (cnfsat);
  \draw[<-]
  (cnfsat)
  to node[sloped,midway,above]
  {\thm\ref{thm:vspcircuitsat}}
  (vspcircuitsat);


  \node (setcover) at (-2,-8) {\setcover};
  \draw[<->,dashed] (hittingset) -- (setcover)
  node[midway, anchor=east] {Open Problem};

  \node (setcovernt) at (-2,-10) {$\setcover/(n+t)$};
  \draw[<->] (setcover) -- (setcovernt)
  node[midway, anchor=west] {\thm\ref{thm:setcovera}};


  \node (steinertree) at (-5,-11.5) {$\steinertree/t$};
  \draw[<-] (setcovernt) -- (steinertree)
  node[sloped,midway,above] {\thm\ref{thm:steiner}};


  \node (cvc) at (-3,-12) {$\cvc/t$};
  \draw[<-] (setcovernt) -- (cvc)
  node[sloped,midway,below] {\thm\ref{thm:cvc}};


  \node (setpart) at (-6,-8.5) {\setpart};
  \draw[<-] (setcover) -- (setpart)
  node[sloped,midway,above]
  {\thm\ref{thm:setpart}};


  \node (subsetsum) at (-6,-10) {$\subsetsum/m$};
  \draw[<-] (setpart) -- (subsetsum)
  node[midway, anchor=east] {\thm\ref{thm:subsetsum}};

  \node (pcnfsat)       at (2,-1.5) {\pcnfsat};
  \node (phittingset)   at (4,-3.25) {\phittingset};
  \node (psetsplitting) at (6,-1.5) {\psetsplitting};
  \node (pnaesat)       at (4, 0.25) {\pnaesat};

  \draw[<-] (pcnfsat)
  to[bend right]
  node[sloped,midway,above]
  {\thm\ref{thm:sattohit}}
  (phittingset);
  \draw[<-] (phittingset)
  to[bend right]
  node[sloped,midway,above]
  {\thm\ref{thm:hit-to-setsplitting}}
  (psetsplitting);
  \draw[<-] (psetsplitting)
  to[bend right]
  node[sloped,midway,below]
  {\obs\ref{obs:setsplitting-to-naesat-to-cnfsat}}
  (pnaesat);
  \draw[<-] (pnaesat)
  to[bend right]
  node[sloped,midway,below]
  {\obs\ref{obs:setsplitting-to-naesat-to-cnfsat}}
  (pcnfsat);

  \draw[<-] (cnfsat) -- (pcnfsat)
  node[pos=.5, anchor=south, align=center]
  {\thm\ref{thm:isolation}\\[-.5\baselineskip]\tiny
  \cite{Calabro_isolation,Traxler}};

  \node[align=center,draw,dashed,label=right:{(\thm\ref{thm:hstosetcover}})] 
  (pallhittingset) at (4,-5.5) {\pallhittingset\\$=$\\\pallsetcover};

  \draw[->] (phittingset) -- (pallhittingset)
  node[sloped,midway,above] {\obs\ref{obs:allhitting to t-hitting}};

  \draw[<-] (pcnfsat)
  to[bend right=40] node[sloped,pos=.8,above] 
  {\thm\ref{thm:psattopallhit}}
  (pallhittingset);


  \node (psetcover) at (6,-8) {\psetcover};
  \draw[<-] (pallhittingset) -- (psetcover)
  node[sloped,midway,below]
  {\obs\ref{obs:allhitting to t-hitting}};

  \node (psetcovernt) at (6,-10) {$\psetcover/(n+t)$};
  \draw[<->] (psetcover) -- (psetcovernt)
  node[pos=.5, anchor=west] {\thm\ref{thm:psetcovera}};


  \node (psteinertree) at (3,-11.5) {$\psteinertree/t$};
  \draw[->] (psteinertree) -- (psetcovernt)
  node[sloped,midway,above] {\thm\ref{thm:steiner}};


  \node (pcvc) at (5,-12) {$\pcvc/t$};
  \draw[->] (pcvc) -- (psetcovernt)
  node[sloped,midway,below] {\thm\ref{thm:cvc}};





  \node (pallhittingsetm) at (2,-8) {$\pallhittingset/m$};
  \node (pcnfsatm) at (2,-10) {$\pcnfsat/m$};

  \draw[<->] (pallhittingset) -- (pallhittingsetm)
  node[sloped,midway,below] {\cor\ref{cor:hs m n}};

  \draw[->]
  (pcnfsatm)
  to node[midway,anchor=west]{\obs\ref{obs:sat m}}
  (pallhittingsetm);

\end{tikzpicture}

  \end{center}
  \caption{%
    Overview of the reduction graph in this paper.
    An arrow $\Pi/s\to\Pi'/s'$ depicts a reduction from the problem 
    $\Pi'$ to the problem $\Pi$ with the following implication:
    If $\Pi$ can be solved in time $c^s\cdot\poly$, then $\Pi'$ can be 
    solved in time $c^{s'}\cdot\poly$.
    The edge labels depict the theorem (T), corollary (C), or 
    observation (O) that contains the formal statement of the 
    reduction.
    When the size parameter $s$ is the number of vertices or variables 
    $n$, we omit it.
    Other parameters are: the number $m$ of clauses, hyperedges, or 
    the number of bits used to represent the input integers in 
    \subsetsum; and the size $t$ of the solution that we are looking 
    for.
    Note that the figure suppresses details about which reductions 
    require or preserve that the instances have bounded clause or 
    hyperedge width, or bounded density.
    On the left, we have decision problems, and on the right we have 
    parity problems; the two groups are related via the isolation 
    lemma~\cite{Calabro_isolation,Traxler}, cf.\ 
    Theorem~\ref{thm:isolation}.
    Furthermore, we observe a cluster on the top, which contains 
    problems for which the best-known algorithm is na\"ive brute 
    force; see Section~\ref{sec:improvbranch}.
    And there is a cluster on the bottom, which contains problems for 
    which the best-known algorithm has a dynamic programming flavor;
    see Section~\ref{sec:dp}.
    These two clusters are connected in the parity world via our 
    ``flip theorem'', Theorem~\ref{thm:hstosetcover}.
    In the decision world, this connection is an open problem: Does 
    SETH imply the assumption of Theorem~\ref{thm:dpsolve}?
  }
\end{figure}
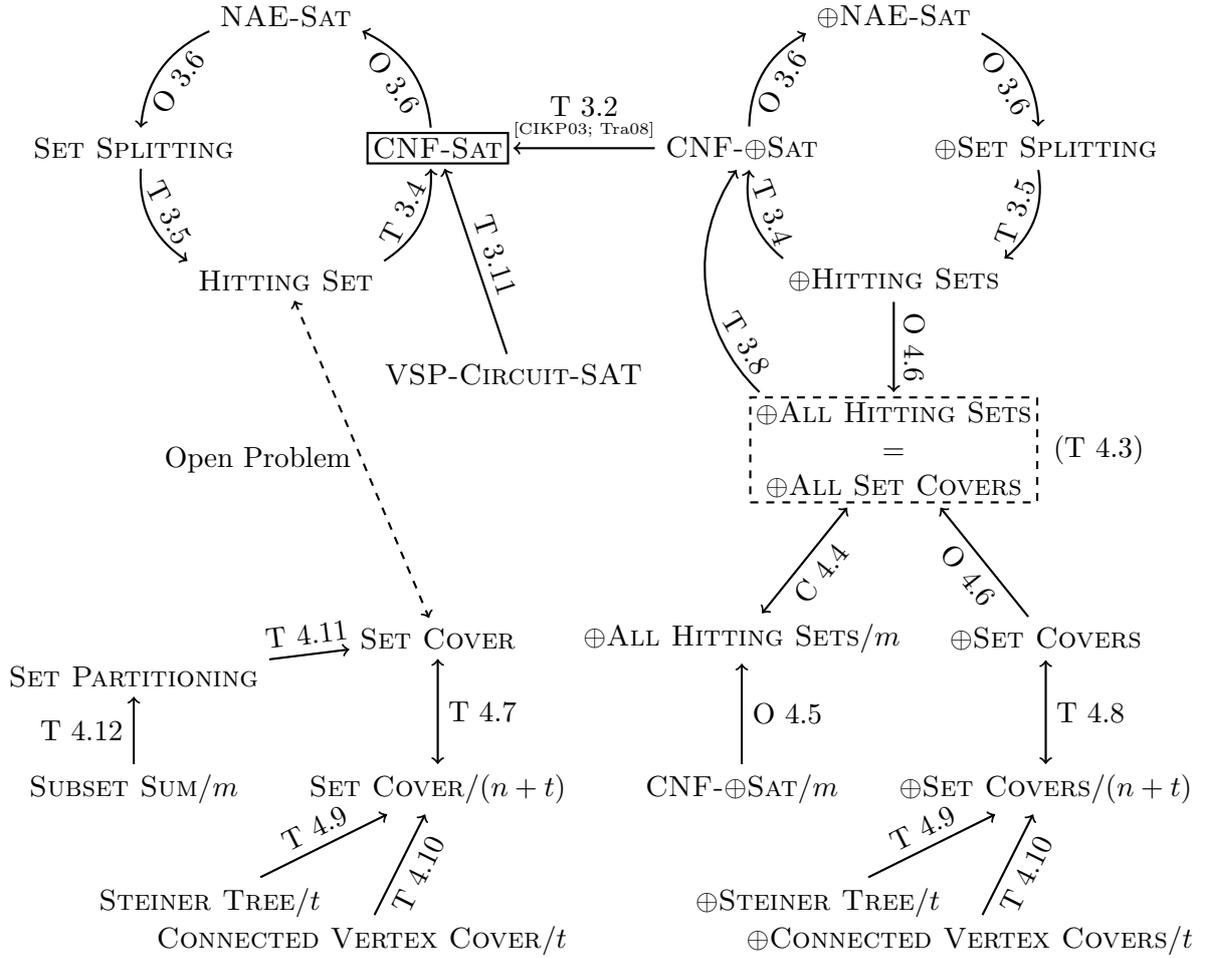

\section{Preliminaries and Notation}\label{sec:prel}

\subsection{General Notation}
In this paper, $\Delta$ denotes the symmetric difference
and~$\mathbin{\dot\cup}$ denotes the disjoint union. For a set~$U$ and
a positive integer $i \leq |U|$, we denote the family of all
subsets of~$U$ of size~$i$ by $\binom{U}{i}$.
In this paper, $\equiv$~will always denote congruence modulo~$2$, that
is, $i \equiv j$ holds for integers $i,j$ if and only if~$i$ and $j$
have the same parity.
Every assignment $\alpha\colon \{v_1,\ldots,v_n\} \to \{0,1\}$ to $n$
Boolean variables $v_1,\ldots,v_n$ is identified with the set
$A := \{v_i \mid \alpha(v_i) = 1\} \subseteq \{v_1,\ldots,v_n\}$.

\subsection{Problem definitions}
\label{sec:problems}

Since we consider a significant number of problems in this paper,
each of which has a few variants, we use the following notation for
clarity.
We write $k$-$\Pi$ for problems whose input consists of set systems of
sets of size at most $k$, or CNF formulas with clauses of width at
most~$k$.
We write $c$-{\sc Sparse}-$k$-$\Pi$ if, in addition, the set systems
or formulas that we get as input are guaranteed to have density at 
most~$c$, that is, the number of sets or clauses is at most $cn$, 
where $n$ is the number of elements or variables.

For each problem~$\Pi$ that we consider, we fix the canonical
NP-verifier that is implicit in the way we define the problem.
Then every yes-instance of $\Pi$ has associated with it a set of
NP-witnesses or ``solutions''.
We write $\oplus\Pi$ for the problem of deciding whether, for a given
instance, the number of its solutions is odd.
For many problems, we are looking for certain subsets of size at 
most~$t$, where $t$ is given as part of the input.
So when writing $\oplus\Pi$ in this case, we only count solutions of 
size at most~$t$.
Sometimes we want to count all solutions, not only those of at most a 
certain size.
In this case, we add the modifier \textsc{All} to the name; for 
example. while \phittingset is the problem of counting modulo two all 
hitting sets of size at most~$t$, the problem \pallhittingset counts 
\emph{all} hitting sets modulo two (regardless of their size).

We now state all problems that we consider in this paper, and we 
discuss how exactly the modifiers affect them.

\subsubsection{CNF Problems}
For CNF problems, the input is a CNF formula~$\varphi$.
We usually denote the number of variables by~$n$ and the number of 
clauses by~$m$.
The two basic problems that we consider are \cnfsat and \naesat.

\begin{description}
  \item[\cnfsat:] Does $\varphi$ have a satisfying assignment?
  \item[\naesat:] Does $\varphi$ have an assignment so that (i) the 
    first variable is set to true and (ii) each clause contains a 
    literal set to true and a literal set to false?
\end{description}
We added condition (i) to \naesat solely for the purpose of making its 
corresponding parity problem non-trivial.

\paragraph{Modifiers.}
In addition to these two basic problems, we can name new problems by 
adding one of the following modifiers to their names (which we do by 
example just for \cnfsat).
\begin{itemize}
  \item \cnfsatargs{$k$} is the problem in which the input 
    formula~$\varphi$ is guaranteed to have at most~$k$ literals in 
    each clause.
  \item \sparsecnfsatargs{$k$}{$c$} is the problem in which the input 
    formula~$\varphi$ is guaranteed to have at most~$k$ literals in 
    each clause and to have at most $m\le c\cdot n$ clauses.
\end{itemize}
The goal of the problem remains the same in both cases, and the two 
modifiers only affect the promise on the input.
In order to change the goal of the problem, we allow for the parity 
modifier, $\oplus$, to be put in front of the type of assignment that 
we are looking for, ie., we have \pcnfsat and \pnaesat.
The parity modifier can be combined with one of the input modifiers.

\subsubsection{Hypergraph Problems}

For problems on hypergraphs, the input is a set system $\mathcal 
F\subseteq 2^U$, which consists of subsets of some universe~$U$.
The elements of $U$ are called \emph{vertices} and the elements of 
$\mathcal F$ are called \emph{hyperedges}.
The number of vertices is usually denoted by~$n$ and the number of 
hyperedges by~$m$.
The goal in all of these problems will be to find or count subsets 
of~$U$ that have special properties with respect to $\mathcal F$, or 
to do the dual and find or count subsets of the set system~$\mathcal 
F$ that have a special property.
Often there will be an additional input $t\in\nats$ that will 
determine that we are looking for a subset $S$ or a subfamily of size 
at most $t$.

We have the following four basic hypergraph problems:

\begin{description}
  \item[\hittingset:]
    Does $\mathcal F$ have a hitting set of size at most $t$, that is, 
    a subset $H \subseteq U$ with $|H| \leq t$ such that $H \cap S 
    \neq \emptyset$ for every $S \in \mathcal{F}$?
  \item[\setcover:]
    Does $\mathcal F$ have a set cover of size at most $t$, that is, a 
    subset $\mathcal{C} \subseteq \mathcal{F}$ with $|\mathcal{C}|\leq 
    t$ such that $\bigcup_{S \in \mathcal{C}}S=U$?
  \item[\setpart \normalfont{(or \textsc{Perfect Set Matching})}:]
    Does $\mathcal{F}$ have a set partitioning of size at most~$t$, 
    that is, a set cover~$\mathcal C$ such that, for every $S,S' \in 
    \mathcal{C}$
    with $S \neq S'$, we have $S \cap S' = \emptyset$?
  \item[\setsplitting:]
    Is there a subset $X\subseteq U$ such that (i) the first element 
    of the universe is a member of $X$ and (ii), for every 
    $S\in\mathcal{F}$, neither $S\subseteq X$ nor $S\subseteq(U-X)$?
\end{description}
Note that the first three problems have the additional input 
$t\in\nats$, while the last problem does not.
Similar to our definition of \naesat, we added condition (i) in 
\setsplitting solely for the purpose of making the corresponding 
parity problem non-trivial.

\paragraph{Modifiers.}
The input modifiers such as in \hittingsetargs{$k$} or 
\sparsehittingsetargs{$k$}{$c$} work as before in the case of CNF 
problems.
The number $k$ promises that all sets $S$ in the set system~$\mathcal 
F$ will have size at most $k$, and the number $c$ promises that the 
number $m$ of sets is at most $c\cdot n$.
We also introduce the parity modifier, $\oplus$, just as before.
For example, in \phittingset, we are given $t$ and $\mathcal F$, and 
we want to count the number of hitting sets of size at most $t$ modulo 
two.

Interestingly, for parity problems, we can prove hardness results also 
for the case in which the input parameter~$t$ is guaranteed to be 
$t=n$.
For decision problems, this setting of~$t$ is trivial, but the 
counting case turns out to be still interesting.
To make this distinction clear, we add the modifier \textsc{All} in 
front of the object that we are counting.
For clarity, we give the definition of the following modified version 
of \hittingset{}:

\vspace{1em}
\noindent
\pallhittingset\vspace{-0.5em}
\begin{description}
  \item[\normalfont{\it Input:}]
    A set system $\mathcal{F} \subseteq 2^{U}$.\vspace{-0.5em}
  \item[\normalfont{\it Question:}]
    Does $\mathcal F$ have an odd number of hitting sets (of any 
    size)?
\end{description}

\subsubsection{Graph Problems}

In graph problems, the input is a graph $G=(V,E)$ with $n$ vertices 
and $m$ edges, and often there is some additional input, such as a 
number $t\in\nats$ or a set of terminals $T\subseteq V$.
We consider the following basic graph problems:
\begin{description}
  \item[\cvc:]
    Does $G$ have a connected vertex cover of size at most $t$, that 
    is, a subset $X \subseteq V$ such that $|X| \leq t$, the induced 
    subgraph $G[X]$ is connected, and $X \cap e \neq \emptyset$ holds 
    for every edge $e \in E$?
  \item[\steinertree:]
    Does $G$ has a Steiner tree of size at most $t$ between the 
    terminals $T\subseteq V$, that is, is there a subset $X\subseteq 
    V$ so that $|X| \leq t$, the induced subgraph $G[X]$ is connected, 
    and $T\subseteq X$?
\end{description}
For these problems, we will only use the parity modifier.
So for example, in \pcvc, we are given $G$ and $t$, and we want to 
count the number of connected vertex covers of size at most $t$ modulo 
two.


\subsubsection{Other Problems}

\vspace{1em}
\noindent
\subsetsum\vspace{-0.5em}
\begin{description}
\item[\normalfont{\it Input:}] Integers $a_1,\ldots,a_n \in 
  \mathbb{Z}_+$
and a target integer $t$ on $m$ bits.\vspace{-0.5em}
\item[\normalfont{\it Question:}] Is there a subset $X \subseteq 
  \{1,\ldots,n\}$ with $\sum_{i \in X}a_i=t$?
\end{description}

\vspace{1em}
\noindent
\vspcircuitsatargs{$c$}\vspace{-0.5em}
\begin{description}
\item[\normalfont{\it Input:}] A $cn$-size Valiant series-parallel 
  circuit over $n$ variables.\vspace{-0.5em}
\item[\normalfont{\it Question:}] Is there a satisfying assignment?
\end{description}

\subsection{The optimal growth rate of a problem}
Running times in this paper have the form $c^n\cdot\poly(m)$,
where~$c$ is a nonnegative constant, $m$~is the total size of the
input, and~$n$ is a somewhat smaller parameter of the input, typically
the number of variables, vertices, or elements.
The constant~$c$ is the \emph{growth rate} of the running time, and it
may be different for different choices for the parameter~$n$.
To make this parameterization explicit, we use the notation $\Pi/n$.
For every such parameterized problem, we now define the
number~$\sigma=\sigma(\Pi/n)$.
\begin{definition}
  For a parameterized problem $\Pi/n$,
  let $\sigma(\Pi/n)$ be the infimum over all $\sigma > 0$ such that
  there exists a randomized $2^{\sigma n}\poly(m)$-time algorithm
  for $\Pi$ whose error probability is at most $1/3$.
\end{definition}
The \emph{optimal growth rate} of $\Pi$ with respect to~$n$ is
$C:=2^{\sigma(\Pi/n)}$.
If the infimum in the definition above is a minimum, then $\Pi$ has an
algorithm that runs in time $C^n\poly(m)$ and no algorithm for $\Pi$
can have a running time $c^n\poly(m)$ for any $c<C$.
On the other hand, if the minimum does not exist, then no algorithm
for~$\Pi$ can run in time $C^n\poly(m)$, but $\Pi$ has a
$c^n\poly(m)$-time algorithm for every $c>C$.
We formally define the Strong Exponential Time Hypothesis (SETH) as 
the assertion that
${\lim_{k\to\infty} \sigma(\cnfsatargs{$k$}/n) = 1}$.

We remark that it is consistent with current knowledge that SETH fails
and yet \cnfsat (without restriction on the clause width) does not 
have $2^{\epsilon n}\poly(m)$-algorithms for any~${\epsilon<1}$:
If SETH fails, then $k$-\cnfsat has, say, $k^{k}1.99^n$-time
algorithms for every~$k$, which does not seem to translate to a
$2^{\epsilon n}\poly(m)$-time algorithm for \cnfsat for any
$\epsilon<1$.

\section{On Improving Branching Algorithms}\label{sec:improvbranch}

In this section we show that significantly faster algorithms for
search problems such as \hittingset{} and \setsplitting{} imply
significantly faster algorithms for \cnfsat{}.
More precisely, we prove that the growth rates of these problems are
equal, or equivalently,
$\sigma(\cnfsat/n)=\sigma(\hittingset{}/n)=\sigma(\setsplitting{}/n)$.
We also give a reduction from \pcnfsat{} to \pallhittingset{}, thus
establishing a connection between the parity versions of these two
problems.

\subsection{Previous results for CNF-SAT}
\label{sec:satprelim}

In the following few subsections, we show reductions from \cnfsat/$n$
to \hittingset/$n$ and \setsplitting/$n$.
These reductions work even when the given instance of \cnfsat/$n$ is
dense, that is, when there is no bound on the number of clauses that
is linear in the number of variables.
However, our starting point in Section~\ref{sec:dp} is the
\sparsehittingset/$n$ problem, where the number of sets in the set
system is linear in~$n$.
For this reason we formulate our results for the sparse versions of
\hittingset/$n$ and \setsplitting/$n$, and we develop a sparse version
of SETH first.

The sparsification lemma by Impagliazzo et al.~\cite{IPZ01} is that
every $k$-CNF formula $\varphi$ can be written as the disjunction of
$2^{\epsilon n}$ formulas in $k$-CNF, each of which has at most
$c(k,\epsilon)\cdot n$ clauses.
Moreover, this disjunction of sparse formulas can be computed
from~$\varphi$ and~$\epsilon$ in time $2^{\epsilon n}\cdot\poly(m)$.
Hence, the growth rate of $k$-\cnfsat for formulas of density at most
$c(k,\epsilon)$ is \mbox{$\epsilon$-close} to the growth rate of general
$k$-\cnfsat.
More precisely,
for every $k$ and every $\epsilon>0$, we have
$\sigma\big(\sparsecnfsatargs{$k$}{$c$}/n\big)
 \leq
 \sigma\big(\cnfsatargs{$k$}/n\big)
 \leq
 \sigma\big(\sparsecnfsatargs{$k$}{$c$}/n\big)+\epsilon$,
where the first inequality is trivial and the second inequality
follows from the sparsification lemma.
The density $c=c(k,\epsilon)$ is the \emph{sparsification constant},
and the best known bound is $c(k,\epsilon)=(k/\epsilon)^{3k}$~\cite{CIP06}.
By setting $\epsilon=\epsilon(k)=o(1)$, this immediately yields
the following theorem.
\begin{theorem}[\cite{IPZ01,CIP06}]\label{thm:sparsificationlemma}
  For every function $c=c(k)\geq(\omega(k))^{3k}$, we have
  \begin{align*}
    \lim_{k \to \infty} \sigma\Big(\cnfsatargs{$k$}/n\Big)
    & =
    \lim_{k \to \infty} \sigma\Big(\sparsecnfsatargs{$k$}{$c$}/n\Big)
    \,.
  \end{align*}
\end{theorem}
Hence, SETH is equivalent to the right-hand side being equal to~$1$.
In \cite{DHMTW12} it was observed that the sparsification lemma can be
made parsimonious, which gives the following equality for the same
functions $c=c(k)$:
  \begin{align*}
    \lim_{k \to \infty} \sigma\Big(\pcnfsatargs{$k$}/n\Big)
    & =
    \lim_{k\to \infty} \sigma\Big(\psparsecnfsatargs{$k$}{$c$}/n\Big)
    \,.
  \end{align*}
We define \pseth{} as the assertion that these limits are equal to~$1$.
The isolation lemmas for $k$-CNF
formulas~\cite{Calabro_isolation,Traxler} immediately yield that
\seth{} implies \pseth{}. More precisely, we have the following
theorem.
\begin{theorem}[\cite{Calabro_isolation,Traxler}]\label{thm:isolation}
  $\lim_{k \to \infty} \sigma(\cnfsatargs{$k$}/n)
  \leq  \lim_{k \to \infty} \sigma(\pcnfsatargs{$k$}/n)\,$.
\end{theorem}

\subsection{From CNF-SAT to Hitting Set}
\label{subsection:cnfsattohset}

Here we will reduce \sparsecnfsat to \sparsehittingset.
For this, and also for the reduction from \pcnfsat to \pallhittingset 
in Section~\ref{subsec:pcnftophs}, the following construction will be 
useful.

Given a CNF formula $\varphi = C_1 \wedge \ldots \wedge C_m$ over $n$
variables $v_1,\ldots,v_n$ and an odd integer $p \ge 3$ that
divides~$n$, we construct the set system
$\mathcal{F}_{\varphi,p} \subseteq 2^{U}$ as follows.
%
%

\begin{enumerate}
	\item
    Let $p'$  be the odd integer $p' = p +  2 \lceil \log_2 p \rceil$,
    and let $U=\{u_1,\ldots,u_{n'}\}$ with $n'=p'\cdot n/p$.
	\item
    Partition the variables of $\varphi$ into
    blocks $V_i$ of size $p$, i.e.,
    $V_i := \{v_{pi+1},\ldots,v_{p(i+1)}\}$.
	\item
    Partition $U$ into blocks $U_i$ of size $p'$, i.e., $U_i =
    \{u_{p'i+1},\ldots,u_{p'(i+1)}\}$.
	\item
    Choose an arbitrary injective function $\displaystyle \psi_i
    \colon 2^{V_i} \rightarrow { \binom{U_i}{\lceil p' / 2 \rceil}}$.
    This exists since
    \[
      \left|\binom{U_i}{\lceil p' / 2 \rceil}\right|
      =
      \binom{p'}{\lceil p' / 2 \rceil}
      \geq
      \frac{2^{p'}}{p'}\geq\frac{2^pp^2}{p+2\lceil\log_2 p\rceil} \geq
      2^p = \left|2^{V_i}
      \right|.
    \]
    We think of $\psi_i$ as a mapping that, given an assignment to the
    variables of $V_i$, associates with it a subset of $U_i$ of size
    $\lceil p' / 2 \rceil$.
  \item
    If $X \in \binom{U_i}{\lceil p' / 2 \rceil}$ for some $i$, then
    add the set $X$ to $\mathcal{F}_{\varphi,p}$.
  \item
    If $X \in \binom{U_i}{\lfloor p' / 2 \rfloor}$ for some $i$ such
    that $\psi^{-1}_i(\{U_i \setminus X\})=\emptyset$, then add the
    set $X$ to $\mathcal{F}_{\varphi,p}$.
  \item
    For every clause $C$ of $\varphi$, do the following:
    \begin{itemize}
      \item
        Let $I = \{ 1 \leq j \leq \frac{n}{p} \mid C \text{ contains a
        variable of block } V_j\}$;
      \item
        For every $i \in I$, we let $\mathcal{A}_i$ be the set
        \begin{equation*}\hspace{-3em}
          \left\{ A \in \binom{U_i}{\lfloor p' /2\rfloor}
          \,\middle|\,
          \text{some assignment in } \psi^{-1}_i(\{U_i\setminus A\})
          \text{ sets all literals in } C\cap V_i \text{ to false}
        \right\};
      \end{equation*}
      \item
        For every tuple $(A_i)_{i \in I}$ with $A_i \in
        \mathcal{A}_i$, add the set $\bigcup_{i\in I}A_i$ to
        $\mathcal{F}_{\varphi,p}$.
    \end{itemize}
\end{enumerate}

\begin{lemma}
  \label{lem:sattohit}
  For every $n$-variable CNF formula~$\varphi$ and every odd integer
  $p\geq 3$ that divides~$n$,
  the number of satisfying assignments of~$\varphi$ is equal to the
  number of hitting sets of size $\lceil\frac{p'}{2}\rceil\frac{n}{p}$
  of the set system~$\mathcal{F}_{\varphi,p}$,
  where $p'=p+2\lceil \log_2 p \rceil$.
\end{lemma}
\begin{proof}
For convenience denote $g=\frac{n}{p}$. Define $\psi \colon 2^V
\rightarrow 2^U$ as $\psi(A) = \bigcup_{i=1}^{g}\psi_i(A \cap V_i)$.
Note that $\psi$ is injective, since for every $i$, $\psi_i$ is
injective. Hence to prove the lemma, it is sufficient to prove that
(1) $A$ is a satisfying assignment if and only if $\psi(A)$ is a
hitting set of size $\lceil\frac{p'}{2}\rceil g$, and
(2) if there is no assignment $A \subseteq V$ such that $\psi(A)=H$,
than no set $H \subseteq U$ of size $\lceil\frac{p'}{2} \rceil g$ is a
hitting set of $\mathcal{F}_{\varphi, p}$.

For the forward direction of (1), note that the sets added in Step 5
are hit by the pigeon-hole principle since $|\psi_i(A \cap
V_i)|=\lceil\frac{p'}{2}\rceil$ and $p'$ is odd.
For the sets added in Step 6, consider the following. The set $X$ of
size $\lfloor p' / 2 \rfloor$ is added because for some $i$,
$\psi^{-1}_i(\{U_i \setminus X\})=\emptyset$. Thus $\psi_{i}(A \cap
V_i)$ automatically hits $X$.
For the sets added in Step 7, consider a clause $C$ of $\varphi$ and
the associated index set $I$ as in Step 7.  Since $A$ is a satisfying
assignment of $\varphi$, there exists $i\in I$ such that $A$ sets at
least one variable in $C\cap V_i$ to true.  Hence, $U_i \setminus
\psi_i(A\cap V_i) \not\in \mathcal{A}_i$.  On the other hand, $U_i
\setminus \psi_i(A\cap V_i)$ is the only member of
$\mathcal{F}_{\varphi,p}$ that cannot be hit by $\psi(A\cap V_i)$.
Therefore, all sets added in Step 7 are hit by $\psi(A)$.
It is easy to check that $\psi(A)$ has size $\lceil\frac{p'}{2}\rceil
g$ since there are $g$ blocks.

For the reverse direction of (1), let $A$ be an assignment such that
$\psi(A)$ is a hitting set of size $\lceil\frac{p'}{2}\rceil g$.  We
show that $A$ is a satisfying assignment of $\varphi$.  Suppose for
the sake of contradiction that a clause $C$ is not satisfied by $A$,
and let $I$ be as defined in Step 7 for this $C$.  Since $\psi(A)$ is
a hitting set, $|\psi(A) \cap U_i| \geq \frac{p'}{2}$ for every $i$
because it hits all sets added in Step 5.  More precisely, $|\psi(A)
\cap U_i| = \lceil\frac{p'}{2}\rceil$ because $|\psi(A)| =
\lceil\frac{p'}{2}\rceil g$ and there are $g$ disjoint blocks
$U_1,\ldots,U_g$.  Therefore, $|U_i \setminus \psi(A)| = \lfloor
\frac{p'}{2} \rfloor$, and so $U_i \cap \psi(A) = U_i \setminus (U_i
\setminus \psi(A))$ is a member of $\mathcal{A}_i$ for every $i$.
This means that in Step 7 the set $\bigcup_{i\in I}A_i$ with $A_i =
U_i \setminus \psi(A)$ was added, but this set is not hit by
$\psi(A)$.  So it contradicts that $\psi(A)$ is a hitting set.

For (2), let $H\subseteq U$ be a set of size $\lceil\frac{p'}{2}
\rceil g$ and assume that there is no assignment $A\subseteq V$ such
that $\psi(A)=H$.  We show that $H$ is not a hitting set of
$\mathcal{F}_{\varphi,p}$.  For the sake of contradiction, suppose
that $H$ is a hitting set.  Then, as in the proof of the reverse
direction of (1), we obtain $|H \cap U_i| = \lceil\frac{p'}{2}\rceil$
for every $i$.  Since it hits all sets added in Step 6, we also know
that $\psi^{-1}_i(\{H\cap U_i\}) \neq \emptyset$ for every $i$.
However, this contradicts the non-existence of $A\subseteq V$ such
that $\psi(A)=H$.
%
\end{proof}

\begin{theorem}\label{thm:sattohit}
  For every non-decreasing function $c=c(k)$, there exists a 
  non-decreasing function $c'=c'(k')$ such
  that
  \begin{align*}
    \lim_{k \to \infty} \sigma(\sparsecnfsatargs{$k$}{$c$}/n)
    &\le
    \lim_{k' \to \infty} \sigma(\sparsehittingsetargs{$k'$}{$c'$}/n)
    \,, \text{ and}\\
    \lim_{k \to \infty} \sigma(\psparsecnfsatargs{$k$}{$c$}/n)
    &\le
    \lim_{k' \to \infty} \sigma(\psparsehittingsetargs{$k'$}{$c'$}/n)
    \,.
  \end{align*}
\end{theorem}
\begin{proof}
We prove that, for any positive integer $k$ and for any positive odd 
integer $p \ge 3$, there exist positive integers $k'=k'(p):=p' k$ and 
$c'=c'(k'):= 2^{k'+1} c(k')$ such that
$$\sigma(\sparsecnfsatargs{$k$}{$c$}/n) \leq \sigma(\sparsehittingsetargs{$k'$}{$c'$}/n) + \OH\left(\frac{\log p}{p}\right)\,.$$
As $p\to\infty$, the right-hand side tends to the right-hand side of 
the inequality that we want to prove, and since the inequality holds 
for all $k$, it also holds as $k\to\infty$.

To prove the claim, we let $\varphi$ be a $k$-CNF formula of density 
at most $c(k)$, and we create the set system $\mathcal{F}_{\varphi,p}$ 
as described above together with the desired hitting set size 
$t=\lceil \frac{p'}{2} \rceil \frac{n}{p}$, and we recall that 
$p'=p+2\lceil \log_2 p\rceil$.
For any constant $p$, this can clearly be done in polynomial time.
By Lemma~\ref{lem:sattohit}, this is a reduction from \cnfsat to 
\hittingset, and the reduction is parsimonious, that is, the number of 
hitting sets is exactly equal to the number of satisfying assignments.
It remains to check that the set system uses at most $c'n'$ sets, each 
of size at most~$k'$, and that the inequality above holds.

It is easy to see that any set in $\mathcal{F}_{\varphi,p}$ has size 
at most $k'$.
Let $m'$ be the number of sets in $\mathcal{F}_{\varphi,p}$.
We observe that there are at most $2^{p'}n/p$ sets added in Step~5 and 
Step~6.
Moreover, since each clause contains variables from at most $k$
blocks, there are at most $2^{p'k}m$ sets added in Step~7.
Therefore $m'/n' \le m'/n \le 2^{p'}+2^{kp'}c(k) \leq c'(k')$ holds, 
where we use the monotonicity of $c$.
This means that we can determine whether $\varphi$ is satisfiable in 
time
$2^{\sigma(\sparsehittingsetargs{$k'$}{$c'$}/n)n'} \cdot\poly(n)$,
where $n'$ is the size of the universe of~$\mathcal{F}_{\varphi,p}$.
Since $n' = \frac{n}{p} (p + 2 \lceil \log p \rceil)=n(1 + 
\OH(\frac{\log p}{p}))$ and $\sigma\leq 1$, the claim follows.
\end{proof}

We remark that the proof also works when there is no restriction on 
the density and even when there is no restriction on the clause/set 
size.
This is because the running time of the reduction is polynomial time 
for every constant $p$.
Furthermore, the theorem trivially holds for the counting versions of 
the problems as well.

\subsection{From Hitting Set via Set Splitting to CNF-SAT}

\begin{theorem}\label{thm:hit-to-setsplitting}
  \begin{align*}
    \lim_{k \to \infty} \sigma(\hittingsetargs{$k$}/n) &\le \lim_{k 
      \to \infty} \sigma(\setsplittingarg{$k$}/n)\,\text{, and}\\
    \lim_{k \to \infty} \sigma(\phittingsetargs{$k$}/n) &\le \lim_{k 
      \to \infty} \sigma(\psetsplittingarg{$k$}/n)\,.
  \end{align*}
\end{theorem}

\begin{proof}
  It is enough to show that, for all positive integers $k$ and $p$, we 
  have
  $$\sigma(\hittingsetargs{$k$}/n) \le 
  \sigma(\setsplittingarg{$k'$}/n)+\frac{\log_2(p+1)}{p}\,,$$
  where $k'=\max(k+1,p+1)$.
  Let $(\mathcal{F},t)$ be an instance of \hittingsetargs{$k$}.
  We can assume that the universe $U$ of $\mathcal F$ has $n$ elements
  and that~$p$ divides $n$.
  Let $U=U_1 \mathbin{\dot\cup} \dots \mathbin{\dot\cup} U_{n/p}$ be a partition in which
  each part has exactly $|U_i|=p$ elements of the universe~$U$.
  Let $t_1,\dots,t_{n/p}$ be nonnegative integers such that
  $\sum_{i=1}^{n/p} t_i = t$.
  The $t_i$'s are our current guess for how many elements of a
  $t$-element hitting set will intersect with the $U_i$'s.
  The number of ways to write $t$ as the ordered sum of $n/p$
  nonnegative integers $t_1,\ldots,t_{n/p}$ with $0\leq t_i \leq p$
  can be bounded by $(p+1)^{n/p}=2^{n \cdot \log (p+1)/p}$.
  For each choice of the $t_i$'s, we construct an instance
  $\mathcal{F'}$ of \setsplittingarg{$k'$} as follows.
  \begin{enumerate}
    \item Let $R$ (red) and $B$ (blue) be two special elements and
      add the set $\{R,B\}$ to $\mathcal{F'}$.
    \item For all $i$ with $t_i<p$ and
      for all $X \in \binom{U_i}{t_i+1}$,
      add $X\cup\{R\}$ to $\mathcal{F'}$.
    \item For every $Y \in \mathcal{F}$,
      add $Y\cup\{B\}$ to $\mathcal{F'}$.
  \end{enumerate}
  Clearly $\mathcal F'$ can be computed in polynomial time and its
  universe has $n+2$ elements.
  The sets added in step~2 have size at most $p+1$ and the sets added
  in step~3 have size at most $k+1$.
  Given an algorithm for \setsplitting, we compute $\mathcal{F'}$ for
  every choice of the $t_i$'s and we decide \hittingset in time
  $2^{(\epsilon+\sigma(\setsplittingarg{$k'$}))\cdot n}\cdot 
  \poly(m)$, where $\epsilon=\log(p+1) /p$.
  It remains to show the correctness of the reduction, i.e., that 
  $\mathcal{F}$ has a hitting set of size at most~$t$ if and only if 
  $\mathcal{F'}$ has a set splitting for some choice of 
  $t_1,\dots,t_{n/p}$.

  For the completeness of the reduction, let $H$ be a hitting set of
  size $t$ and set $t_i=|U_i\cap H|$ for all~$i$.
  We now observe that $H\cup\{R\}$ and its complement $(U-H)\cup\{B\}$ form a set
  splitting of $\mathcal F'$.
  The set $\{R,B\}$ added in step~1 is split.
  The sets $X\cup\{R\}$ added in step~2 are split since at least one
  of the $t_i+1$ elements of $X\subseteq U_i$ is not contained in $H$.
  Finally, the sets $Y\cup\{B\}$ added in step~3 are split since each
  $Y\in\mathcal{F}$ has a non-empty intersection with $H$.

  For the soundness of the reduction, let $(S,\overline S)$ be a set
  splitting of $\mathcal F'$ for some choice of $t_1,\ldots,t_{n/p}$.
  Without loss of generality, assume that $R$ is the first vertex and 
  thus, because of the way we defined \setsplitting, we will have 
  $R\in S$.
  By the set added in step~1, this means that $B\in\overline{S}$.
  The sets added in step~2 guarantee that $U_i\cap S$ contains at
  most~$t_i$ elements for all~$i$.
  Finally, the sets added in step~3 make sure that each set
  $Y\in\mathcal{F}$ has a non-empty intersection with~$S$.
  Thus, $S\setminus \{R\}$ is a hitting set of $\mathcal{F}$ and has size at
  most $\sum_i t_i=t$.

  The claim for the parity versions follows as well since the 
  reduction preserves the number of solutions exactly.
\end{proof}

\begin{observation}\label{obs:setsplitting-to-naesat-to-cnfsat}
For any positive integer $k$ we have
\begin{align*}
  \sigma(\setsplittingarg{$k$}/n)
  &\leq \sigma(\naesatargs{$k$}/n)
  \leq \sigma(\cnfsatargs{$k$}/n)\,\text{, and}\\
  \sigma(\psetsplittingarg{$k$}/n)
  &\leq \sigma(\pnaesatargs{$k$}/n)
  \leq \sigma(\pcnfsatargs{$k$}/n)\,.
\end{align*}
\end{observation}
\begin{proof}
  For the first reduction, let $\mathcal F$ be an instance of
  \setsplittingarg{$k$}.
  We construct an equivalent $k$-CNF formula $\varphi$ as follows.
  For each element in the universe of $\mathcal F$, we add a variable,
  and for each set $X\in\mathcal F$ we add a clause in which each
  variable occurs positively.
  A characteristic function of a set splitting $U=U_1 \mathbin{\dot\cup} U_2$ is one that assigns
  $1$ to the elements in $U_1$ and
  $0$ to the elements of $U_2$.
  Observe that the characteristic functions of set splittings of~$\mathcal F$
  stand in one-to-one correspondence to variable assignments that
  satisfy the \naesat constraints of~$\varphi$.
  Thus, any algorithm for \naesatargs{$k$} works for \setsplittingarg{$k$}, too.

  For the second reduction, let $\varphi$ be a \naesatargs{$k$}-formula.
  The standard reduction to \cnfsatargs{$k$} creates two copies of every clause
  of $\varphi$ and flips the sign of all literals in the second
  copies.
  Then any \naesat-assignment of $\varphi$ satisfies both copies of
  the clauses of $\varphi'$.
  On the other hand, any satisfying assignment of $\varphi'$ sets a
  literal to true and a literal to false in each clause of $\varphi$.
  To make the satisfying assignments of $\varphi'$ exactly the same as 
  the NAE-assignments of $\varphi$, we furthermore add a single clause 
  that forces the first variable of $x$ to be set to true (recall that 
  this requirement was part of our definition of \naesat).
  Thus, any algorithm for \cnfsatargs{$k$} works for \naesatargs{$k$}, 
  too.
\end{proof}

\subsection{From Parity CNF-SAT to Parity All Hitting 
  Sets}\label{subsec:pcnftophs}

Given a CNF formula $\varphi$ over $n$ variables and clauses of size 
at most $k$ and an odd integer $p \ge 3$ that divides $n$, we first 
construct the set system $\mathcal{F}_{\varphi,p} \subseteq 2^{U}$ as 
described in Section~\ref{subsection:cnfsattohset}.
Given the set system $\mathcal{F}_{\varphi,p} \subseteq 2^{U}$, we 
create the set system $\mathcal{F}'_{\varphi,p}$ as follows:

\begin{enumerate}
	\setcounter{enumi}{7}
	\item For every block $U_i$:
		\begin{itemize}
			\item add a special element $e_i$ to the universe,
      \item for every $X \in \binom{U_i}{\lfloor p'/2 \rfloor}$, add the set $X \cup \{ e_i \}$ to the set family.
		\end{itemize}
\end{enumerate}

\begin{lemma}
\label{lem:psattophit}
The number of hitting sets of size $t=\lceil p'/2 \rceil\frac{n}{p}$ 
in $\mathcal{F}_{\varphi,p}$ is odd
if and only if
the number of \emph{all} hitting sets in $\mathcal{F}'_{\varphi,p}$ is 
odd.
\end{lemma}
\begin{proof}
Let $g=\frac{n}{p}$. We first prove that the number of hitting sets of $\mathcal{F}_{\varphi,p}$ of size $\lceil p'/2 \rceil g$ is equal to the number of hitting sets $H'$ of $\mathcal{F}'_{\varphi,p}$ such that $|H' \cap U_i| = \lceil \frac{p'}{2} \rceil$ for every $1 \leq i \leq g$. Suppose that $H$ is a hitting set of $\mathcal{F}_{\varphi,p}$ of size $\lceil p'/2 \rceil g$, then it is easy to see that $H \cup \{e_1,\ldots,e_g\}$ is a hitting set of $\mathcal{F}'_{\varphi,p}$ since all the sets added in Step 8 are hit by some $e_i$, and indeed  $|H' \cap U_i| = \lceil \frac{p'}{2} \rceil$ for every $1 \leq i \leq g$ since otherwise the set $U_i \setminus H'$ added in Step 5 is not hit by $H'$.
For the reverse direction, suppose $H'$ is a hitting set of $\mathcal{F}'_{\varphi,p}$ such that $|H' \cap U_i| = \lceil \frac{p'}{2} \rceil$ for every $1 \leq i \leq g$.
Then $\{e_1,\ldots,e_g\} \subseteq H'$ since all the sets added in Step 8 are hit by $H'$. And hence we have a bijection between the two families of hitting sets.

For every hitting set $H'$ of $\mathcal{F}'_{\varphi,p}$ and block
$U_i$, we know that $|H' \cap U_i| \geq \lceil p'/2 \rceil$. So it
remains to show that the number of hitting sets $H'$ of
$\mathcal{F}'_{\varphi,p}$ such that there is an $1 \leq i \leq g$
with $|H' \cap U_i| > \lceil \frac{p'}{2}\rceil$ is even. Given such a
hitting set $H'$, let $\gamma(H')=H' \Delta \{e_i\}$ where $i$ is the
smallest integer such that $|H' \cap U_i| > \lceil
\frac{p'}{2}\rceil$. Obviously $\gamma$ is its own inverse and
$|\gamma(H') \cap U_i| > \lceil \frac{p'}{2}\rceil$ so now it remains
to show that $\gamma(H')$ is also a hitting set of
$\mathcal{F}'_{\varphi,p}$. To see this, notice that all sets $X \cup
\{ e_i \}$ added in Step 8 where $X \in \binom{U_i}{\lfloor p' /2 \rfloor}$ are hit since $|\gamma(H') \cap U_i| > \lceil \frac{p'}{2}\rceil$ and that those are the only sets containing $e_i$.
\end{proof}

\begin{theorem}\label{thm:psattopallhit}
  For every non-decreasing function $c=c(k)$, there exists a 
  non-decreasing function $c'=c'(k')$ such that
  \begin{align*}
    \lim_{k \to \infty} \sigma(\psparsecnfsatargs{$k$}{$c$}/n)
    &\le
    \lim_{k' \to \infty} 
    \sigma(\pallsparsehittingsetargs{$k'$}{$c'$}/n)
    \,.
  \end{align*}
\end{theorem}
\begin{proof}
  Let $\varphi$ be an instance of \psparsecnfsatargs{$k$}{$c$}.
  First recall from the proof of Theorem~\ref{thm:sattohit} that the 
  reduction
  \begin{align*}
    \sigma(\psparsecnfsatargs{$k$}{$c$}/n)
    &\le
    \sigma(\psparsehittingsetargs{$k'$}{$c'$}/n)
    + O\left(\frac{\log p}{p}\right)
  \end{align*}
  worked by constructing the set system $\mathcal F_{\varphi,p}$, and 
  that the reduction was parsimonious.
  Thus, when we now further move to $\mathcal F'_{\varphi,p}$, we have 
  that the parity of the number of all hitting sets in $\mathcal F_{\varphi,p}'$ 
  is equal to the parity of the number of hitting sets of size at 
  most~$t$ in $\mathcal F_{\varphi,p}$ (by 
  Lemma~\ref{lem:psattophit}), which in turn is equal to the parity of 
  the number of satisfying assignments to $\varphi$.
  Thus, this is a valid reduction from \pcnfsat to \pallhittingset; 
  since the maximum edge size $k'$ does not increase, we just have to 
  verify that the instance remains sparse and does not have too many 
  more vertices.

  For the density, note that, in Step 8, we add at most $2^{p'}n/p$ 
  sets, so the density $c'$ of $\mathcal F_{\varphi,p}$ goes up by at 
  most an additive term of $2^{p'}/p$, which can be easily bounded by 
  a function just of~$k'$.
  For the running time, note that the number $n'$ of vertices in 
  $\mathcal F_{\varphi,p}$ goes up by exactly $n/p'$, that is, the new 
  number $n''$ of vertices can be bounded by $n''\le (1+1/p') n'$.
  As $p\to\infty$, this will approach $n'' \leq n'$.
  The claim follows because we can determine the parity of the number 
  of hitting sets of size at most~$t$ in the set system 
  $\mathcal{F}_{\varphi,p}$ by running the best algorithm for the 
  corresponding problem \pallhittingset, which runs in time
  $2^{\sigma(\pallsparsehittingsetargs{$k'$}{$c''$}/n)n''} \cdot 
  \poly(m)$.
\end{proof}

Note that conversely, an improved algorithm for \pcnfsat gives an 
improved algorithm for \pallhittingset.
This is because instances of \pallhittingset can be viewed in a 
natural way a monotone CNF formulas:
given a set family $\mathcal{F} \subseteq U$ we simply associate a 
variable with every element of $U$ and a monotone clause for every
set $S \in \mathcal{F}$.
\begin{observation}\label{observation:monotone CNF Sat}
  For all positive integers $k$ and $c$, we have
  \[
    \sigma(\pallsparsehittingsetargs{$k$}{$c$}/n)
    \leq
    \sigma(\psparsecnfsatargs{$k$}{$c$}/n)
  \]
\end{observation}


\subsection{Satisfiability for Series-Parallel Circuits}
In this subsection, we show that the satisfiability of $cn$-size {\em series-parallel} circuits can be decided in time time $2^{\delta n}$ for $\delta <1$ independent of $c$ if and only if $\seth$ is not true.
Here the size of a circuit is the number of wires.
Our proof is based on a result of Valiant regarding paths in sparse graphs \cite{Valiant_1977_mfcs}. Calabro \cite{Calabro_2008_eccc} discusses various notions of series-parallel graphs and provides a more complete proof of Valiant's lower bound on the size of series-parallel graphs (which he calls Valiant series-parallel graphs) that have ``many'' long paths.
We remark that the class of Valiant series-parallel graphs is not the same as the notion of series-parallel graphs 
used most commonly in graph theory (see \cite{Calabro_2008_eccc}).

In this section a {\em multidag} $G=(V,E)$ is a directed acyclic multigraph. Let $\inp(G)$ denote the set of vertices $v \in V$ such that the indegree of $v$ in $G$ is zero. Similarly, let $\out(G)$ denote the set of vertices $v \in V$ such that the outdegree of $v$ in $G$ is zero. A {\em labeling} of $G$ is a function $l\colon V \rightarrow \nats$ such that $\forall (u,v)\in E$, $l(u)<l(v)$. A labeling $l$ is {\em normal} if for all $v \in \inp(G)$, $l(v)=0$ and there exists an integer $d\in \nats$ such that for all $v \in \out(G) \setminus \inp(G)$,  $l(v)=d$. A multidag $G$ is \emph{Valiant series-parallel} ($\vsp$) if it has a normal labeling $l$ such that
there exist no $(u,v), (u',v')\in E$ such that $l(u)<l(u')<l(v)<l(v')$.

We say that a boolean circuit $C$ is a \emph{$\vsp$ circuit} if the underlying multidag of $C$ is a $\vsp$ graph and the indegree of every node is at most two
(namely, the fan-in of each gate is at most two).
Using the depth-reduction result by Valiant~\cite{Valiant_1977_mfcs}
and following the arguments by Calabro~\cite{Calabro_2008_eccc} and
Viola~\cite{Viola_2009_FTTCS},
we may show the following.

\begin{theorem}
\label{thm:depth3}
Let $C$ be a $\vsp$ circuit of  size $cn$ with $n$ input variables.
There is an algorithm $A$ which on input $C$ and a parameter $d\geq 1$
outputs an equivalent depth-3
unbounded fan-in OR-AND-OR circuit $C'$ with the following properties.
\begin{enumerate}
\item Fan-in of the top OR gate in $C'$ is bounded by $2^{n/d}$.
\item Fan-in of the bottom OR gates is bounded by $2^{2^{\mu c d}}$ where
$\mu$ is an absolute constant.
\item 
$A$ runs in time $O(2^{n/d} n^{O(1)})$ if $c$ and $d$ are constant.
\end{enumerate}
\end{theorem}

In other words, for all $d\geq 1$, Theorem \ref{thm:depth3} reduces the satisfiability of a $cn$-size $\vsp$ circuit to that of the satisfiability of a disjunction of $2^{n/d}$ $k$-CNFs where $k\leq 2^{2^{\mu cd}}$
in time $O(2^{n/d} n^{O(1)})$.
This implies that
\[
\sigma(\vspcircuitsatargs{$c$}/n) \leq \sigma(\cnfsatargs{$2^{2^{\mu c d}}$}/n) + \frac{1}{d}.
\]
Hence, we obtain the following theorem.

\begin{theorem}\label{thm:vspcircuitsat}
\[
\lim_{c\to \infty} \sigma(\vspcircuitsatargs{$c$}/n)
\leq
\lim_{k\to \infty} \sigma(\cnfsatargs{$k$}/n).
\]
\end{theorem}

For the reverse direction,
observe that a CNF formula with $cn$ clauses, all of size at most $k$, can be written as a $4ck$-size $\vsp$ circuit. This observation implies that
\[
\sigma(\sparsecnfsatargs{$k$}{$c$}/n) \leq \sigma(\vspcircuitsatargs{$4ck$}/n).
\]
Together with the sparsification lemma,
Theorem~\ref{thm:sparsificationlemma}, we obtain the following
theorem.

\begin{theorem}
  $
  \lim_{k\to \infty} \sigma(\cnfsatargs{$k$}/n)
  \leq
  \lim_{c\to \infty} \sigma(\vspcircuitsatargs{$c$}/n)
  \,.
  $
\end{theorem}

\section{On Improving Dynamic Programming Based Algorithms}
\label{sec:dp}

In this section we give some reductions that show that several dynamic 
programming based algorithms cannot be improved unless the growth rate 
of \cnfsat can be improved.
In the parity world, our starting point will be the hardness of 
$\pallhittingset/n$ as proved in Theorem~\ref{thm:psattopallhit}. More 
specifically, we show that $\pallhittingset$ and $\pallsetcover$ are 
actually the \emph{same} problem, for which we use a simple but novel 
property of independent sets in bipartite graphs in 
\S\ref{subsec:flip}.
In \S\ref{subsec:smallsc} we show that the current algorithms for 
$\psteinertree/t$ and $\pcvc/t$ are at least as hard to improve as the 
algorithm for $\pallsetcover/n$.
Motivated by these facts, we concoct the hypothesis that the growth 
rate $2$ of the best known algorithm for \setcover can not be 
improved, and we show similar implications for the problems 
$\steinertree/t$ and $\cvc/k$, \setpart and \subsetsum.

\subsection{The flip: Parity Hitting Set equals Parity Set Cover}
\label{subsec:flip}

It is well known that the Hitting Set and the Set Cover problem are 
dual to each other:
The hitting sets of any set family $\mathcal F$ are in one-to-one 
correspondence with the set covers of its dual set family $\mathcal 
F^*$.
Here the dual is defined by flipping the roles of sets and elements: 
in $\mathcal F^*$, every element becomes a set and every set becomes 
an element, but we preserve all incidences between them.
\begin{observation}\label{observation:weak duality}
  For all set families $\mathcal F$, we have
  \[\pallhittingset(\mathcal F) = \pallsetcover(\mathcal F^*)\,.\]
\end{observation}

We demonstrate now that, in the parity world, the duality between 
hitting set and set cover is very strong: Indeed, the two parities are  
equal even without going to the dual set system!  For this, we first 
state the following intermediate step.
\begin{lemma}
  \label{lem:theflip}
  Let $G=(A \cup B,E)$ be a bipartite graph, then the number of 
  independent sets of~$G$ modulo two is equal to
  $
  |\{ X \subseteq A: N(X) = B \}| \bmod 2
  $.
\end{lemma}
\begin{proof}
Grouping on their intersection with $A$, the number of independent sets of $G$ is equal to
\begin{align*}
 \sum_{X \subseteq A} 2^{|B \setminus N(X)|}
 &\equiv
 \sum_{\substack{X \subseteq A \\ |B \setminus N(X)|=0}}2^0
 \equiv
 |\{ X \subseteq A: N(X) = B \}| \bmod 2\,.
 &
 \tag*{\qedhere}
 \end{align*}
\end{proof}

This lemma was inspired by a non-modular variant from \cite[Lemma 
2]{DBLP:conf/iwpec/NederlofR10} (see also \cite[Proposition 
9.1]{phdthesisrooij}).
We now show that, for any set system, the parity of the number of 
hitting sets is always equal to the parity of the number of set 
covers.
\begin{theorem}[Flip Theorem]
\label{thm:hstosetcover}
$
\pallhittingset=\pallsetcover\,
$.
\end{theorem}
\begin{proof}
Let $\mathcal{F} \subseteq 2^{U}$ be a set system, let 
$G=(\mathcal{F},U,E)$ be the bipartite graph where $(S,e) \in E$ if 
and only if $e \in S$.  Note that the number of hitting sets of 
$\mathcal{F}$ is equal to $|\{X \subseteq U: N(X) = \mathcal{F}\}|$.  
Then by Lemma \ref{lem:theflip}, the number of hitting sets is equal 
to the number of independent sets of $G$ modulo 2. And similarly, 
since the lemma is symmetric with respect to the two color classes of 
the bipartite graph, the number of set covers of $\mathcal{F}$ is also 
equal to the number of independent sets of $G$ modulo 2.
Thus all three parities are equal.
\end{proof}

Let us emphasize once again that the problem \pallhittingset is equal 
to the problem \pallsetcover.
If, in the following, we use two different names, we do so only 
because the view of one or the other is more convenient for us.

The duality observation and the theorem above give rise to the 
following curious corollary.
\begin{corollary}\label{cor:hs m n}
  $\sigma(\pallhittingset/n) = \sigma(\pallhittingset/m)$
\end{corollary}
That is, \pallhittingset has a $1.99^n \cdot \poly(m+n)$ algorithm if 
and only it has a $1.99^m \cdot \poly(m+n)$ algorithm.
Since hitting sets can be seen as satisfying assignments of a monotone 
CNF formula, we can also formulate an analogue of 
Observation~\ref{observation:monotone CNF Sat}.
\begin{observation}\label{obs:sat m}
  $\sigma(\pallhittingset/m) \leq \sigma(\pcnfsat/m).$
\end{observation}
Putting all things together, we proved that a $1.99^m\cdot\poly(m+n)$ 
algorithm for \pcnfsat implies a $1.99^n\cdot\poly(m+n)$ time 
algorithm for the same problem, and thus such an algorithm would 
violate \seth.

We finish this discussion with one more observation: We can always 
reduce from the problem \pallhittingset to \phittingset and to 
\psetcover.
\begin{observation}\label{obs:allhitting to t-hitting}
  For all size parameters~$s$ of $\pallhittingset$, we have
  \begin{align*}
    \sigma(\pallhittingset/s) &\leq \sigma(\phittingset/s)\,\text{, 
      and}\\
    \sigma(\pallhittingset/s) &\leq \sigma(\psetcover/s)\,.
  \end{align*}
\end{observation}
\begin{proof}
  Note that $\pallhittingset$ is equal to the problem $\phittingset$ 
  in which the size $t$ of the hitting sets we are counting is fixed 
  to $t=n$, i.e., we count all hitting sets.
  Then any algorithm for $\phittingset$ will immediately work for 
  $\pallhittingset$ as well.
  The analogous argument applies to $\psetcover$.
\end{proof}

\subsection{From Set Cover to Steiner Tree and Connected Vertex Cover}
\label{subsec:smallsc}

In this subsection we will give reductions from \setcover/$n$ to
\steinertree/$t$ and \cvc/$k$.
We transfer the reductions to the parity versions \setcover/$n$,
\psteinertree/$t$, and \pcvc/$k$.
For the reduction, we first need an intermediate result, showing that
\setcover/$(n+t)$, that is, \setcover parameterized by the sum of the
size of the universe and solution size, is as hard as \setcover/$n$
(and similarly for \psetcover/$n$ and \psetcover/$(n+t)$).
Once we have this intermediate result, the reductions to the
\psteinertree/$t$ and \pcvc/$k$ problems follow more easily.



\begin{theorem}
\label{thm:setcovera}
$
  \lim_{k \to \infty} \sigma(\setcoverarg{$k$}/n)
  =
  \lim_{k \to \infty} \sigma(\setcoverarg{$k$}/(n+t))
  \,.
$
\end{theorem}
\begin{proof}
  The case $\ge$ follows from the basic fact that increasing the size 
  parameter cannot increase the running time relative to the 
  parameter.

  To prove $\le$, we use the ``powering'' technique for \setcover:
  for each constant $\alpha > 0$, we transform an instance $(\cF,U,t)$ 
  of \setcoverarg{$k$}
  into 
  an instance of $\setcoverarg{$k'$}$,
  for some positive integer $k'$,
  where the size  $t'$ of the solution in the resulting $\setcoverarg{$p'$}$ instances
  is at most $\alpha |U|$, without changing the universe size.

Without loss of generality, we assume that $t \le |U|$.
Consider any $\alpha > 0$.
Let $q$ be the smallest positive integer such that $\frac{1}{q} \le \alpha$.
We may assume that $t$ is divisible by $q$, since otherwise we may add
at most $q$ additional elements to the universe $U$ and singleton sets to the family $\cF$.
We form a family $\cF'$ of all unions of
exactly $q$ sets from $\cF$, that is for each of $\binom{|\cF|}{q}$ choices of $q$
sets $S_1,\ldots,S_q \in \cF$ we add to $\cF'$ the set $\bigcup_{i=1}^q S_i$.
Note that since $q$ is a constant we can create $\cF'$ in polynomial time.
We set $t'=t/q \le |U|/q \le \alpha|U|$.
It is easy to see that $(\cF,U,t)$ is a YES-instance of $\setcoverarg{$k$}$
if and only if $(\cF',U,t')$ is a YES-instance of $\setcoverarg{$qk$}$.
\end{proof}

Observe that in the proof above, because of the grouping of $q$ sets,
one solution for the initial instance may correspond to several
solutions in the resulting instance. For this reason
the counting variant of the above reduction is much more technically involved.

\begin{theorem}
  \label{thm:psetcovera}
  For every function $c=c(k)$, we have
  \begin{align*}
    \lim_{k \to \infty} \sigma(\psparsesetcoverargs{$k$}{$c$}/n)
    &\le
    \lim_{k' \to \infty} \sigma(\psetcoverarg{$k'$}/(n+t))\,.
  \end{align*}
\end{theorem}
The reverse $\sigma(\psparsesetcoverargs{$k$}{$c$}/n) \ge 
\sigma(\psparsesetcoverargs{$k$}{$c$}/(n+t))$ holds trivially for all 
$k$ and $c$.
The proof of Theorem~\ref{thm:psetcovera} is quite involved, and we 
postpone it to the end of this section.
Instead, we will first look at some of its consequences.
\begin{theorem}
  \label{thm:steiner}
  \begin{align*}
    \lim_{k \to \infty} \sigma(\setcoverarg{$k$}/(n+t))
    &\le
    \sigma(\steinertree/t)
    \,, \text{ and}\\
    \lim_{k \to \infty} \sigma(\psetcoverarg{$k$}/(n+t))
    &\le
    \sigma(\psteinertree/t)
    \,.
  \end{align*}
\end{theorem}
\begin{proof}
Given an instance of \setcover consisting of a set system $(\mathcal{F},U)$ and integer $i$, let $G'$ be the graph obtained from the incidence graph of $(\mathcal{F},U)$ by adding a vertex $s$ universal to $\mathcal{F}$ with a pendant vertex $u$, and define the terminal set to be $U \cup \{ u \}$. It is easy to see that the number of Steiner trees with $|U|+i+1$ edges is equal to the number of set covers of $(\mathcal{F},U)$ of size $i$. Hence the theorem follows.
\end{proof}

\begin{theorem}
\label{thm:cvc}
\begin{align*}
\lim_{k \to \infty} \sigma(\setcoverarg{$k$}/(n+t))
& \le \sigma(\cvc/t)\,, \text{ and} \\
\lim_{k \to \infty} \sigma(\psetcoverarg{$k$}/(n+t))
& \le \sigma(\pcvc/t)\,.
\end{align*}
\end{theorem}
\begin{proof}
Given an instance $(\mathcal{F},U,t)$ of \setcover, we create an
instance of \cvc with $G$ being obtained from the incidence graph of
$(\mathcal{F},U)$ by adding a vertex $s$ adjacent to all vertices
corresponding to sets and adding pendant vertices for every element of
$U \cup \{s\}$.
Moreover let $t'=t+|U|+1$ in the \cvc instance.

It is easy to see that for every $i$, there exists a set cover of
$(\mathcal{F},U)$ of size $i \leq t$ if and only if there exists a
connected vertex cover of $G$ of size at most $i+|U|+1\leq t'$ since
we can take without loss of optimality all vertices having a pendant
vertex, and then connecting these vertices is equivalent to covering
all elements of $U$ with sets in $\mathcal{F}$. Hence, by using an
algorithm for \cvc, we obtain an
$\OH(2^{\sigma(\cvc/t)t'}n^{O(1)})=\OH(2^{\sigma(\cvc/t)(|U|+t)}n^{O(1)})$
time algorithm for \setcoverarg{$p$}.

For the parity case, let us study the number of connected vertex
covers of size $j$ of $G$ for every $j$. Similarly to the previous
case, note that for any connected vertex cover $C$, $C \cap
\mathcal{F}$ must be a set cover of $(\mathcal{F},U)$ by the
connectivity requirement. Hence we group all connected vertex covers
in $G$ depending on which set cover in $(\mathcal{F},U)$ their
intersection with $\mathcal{F}$ is. Let $c_j$ be the number of
connected vertex covers of $G$ of size $j$ and $s_i$ be the number of
set covers of size $i$ in $(\mathcal{F},U)$, then:
\begin{align*}
	c_j
  &=
  \sum_{i=1}^{j-|U|-1} s_i \binom{|U|+1}{j-i-|U|-1}
  \,.
\end{align*}
Now the number $s_i$ modulo $2$ can be determined in polynomial time
once $(c_1,\ldots,c_{i+|U|+1})$ modulo $2$ are computed by recovering
$s_1$ up to $s_i$ in increasing order, since for $i=j-|U|-1$ we have
$\binom{|U|+1}{j-i-|U|-1}=1$.

Thus, if in time $\OH(2^{\sigma(\cvc/t)t'}n^{O(1)})$  we can compute
the number of connected vertex covers of size $n$ modulo $2$, we can
compute the parity of all $(c_1,\ldots,c_{i+|U|+1})$ and hence the
parity of $s_i$ in $\OH(2^{\sigma(\cvc/t)(|U|+t)}n^{O(1)})$.
\end{proof}

\subsection{From Set Cover via Set Partitioning to Subset Sum}

\begin{theorem}\label{thm:setpart}
\[
\lim_{p \to \infty} \sigma(\setcoverarg{$p$}/n) \leq \lim_{p \to \infty} \sigma(\setpartarg{$p$}/n).
\]
\end{theorem}
\begin{proof}
Let~$(\mathcal{F},t)$ be an instance of \setcoverarg{$p$}.
Create an instance~$(\mathcal{F'},t)$ of \setpartarg{$p$} by
for every~$S \in \mathcal{F}$ adding all subsets of~$S$ to~$\mathcal{F'}$.  Clearly~$(\mathcal{F'},t)$ has a
set partitioning of size at most~$t$ if and only if~$(\mathcal{F},t)$ has a set cover of size at most~$t$.
Since the size of the sets in~$\mathcal{F}$ is bounded by~$p$, the reduction runs in polynomial time.
\end{proof}

\begin{theorem}\label{thm:subsetsum}
\[
\lim_{k \to \infty} \sigma(\setpartarg{$k$}/n) \leq \sigma(\subsetsum/m).
\]
\end{theorem}
\begin{proof}
Let~$\mathcal{F}\subseteq 2^U$ be an instance of \setpartarg{$k$}.
We iterate over all potential sizes $1 \le t_0 \le n$ of the solution for the \setpart problem.

We create an instance of \subsetsum as follows.
Let the target integer~$t$ for \subsetsum have a bit expansion consisting of three fields.
First, as the most significant bits, a field coding the value of~$t_0$, to check the cardinality of the solution $\mathcal{C} \subseteq \mathcal{F}$;
second, a field of length~$\log_2 t_0+\log_2 n$ containing the value~$n$, to check the total size of all sets in~$\mathcal{C}$;
finally, a field of length~$\log_2 t_0+n$ containing~$n$ ones.  The paddings of length~$\log_2 t_0$ serve to isolate the fields from each other.
For every~$S_i \in \mathcal{F}$, we create an integer~$a_i$ with the same field division as~$t$, where the first field encodes~$1$, the second field encodes~$|S_i|$,
and the third field contains a one in position~$j$ if and only if~$u_j \in S_i$.  We argue that the resulting \subsetsum instance is a YES-instance
if and only if~$\mathcal{F}$ contains a partitioning of~$U$ using exactly~$t_0$ sets.

Clearly, if~$\mathcal{C} \subseteq \mathcal{F}$ partitions~$U$ and~$|\mathcal{C}|=t_0$, then the integers~$a_i$ corresponding to~$S_i \in \mathcal{C}$ sum to~$t$.
The first field sums to~$t_0$ by cardinality of~$\mathcal{C}$, the second sums to~$n$, and in the third field the non-zero digits are simply
partitioned between the~$a_i$.

So let~$A$ be a collection of integers~$a_i$ that sum to~$t$.  By the first field, we have~$|A| \leq t_0$; thus the padding of length~$\log t_0$
is enough to isolate the fields, and we have~$|A|=t_0$.  By the same argument on the second field, the sum over all~$a_i\in A$ of the number of
non-zero bits in the third field is exactly~$n$.  Now, the only way that the third field can actually contain~$n$ true bits is if the
true bits in the third field are partitioned among the~$a_i$.  Thus,~$\mathcal{C}=\{S_i \mid a_i \in A\}$ is a set partitioning of~$U$
of cardinality exactly~$t_0$.

By looping over all~$1 \le t_0 \leq t$ for the \setpart instance, this solves the problem.
Note that the length of the bit string~$t$ is~$n+\OH(\log n)$, which disappears into the asymptotics.
\end{proof}

\subsection{Proof of Theorem~\ref{thm:psetcovera}}



As a proof we present a reduction which
for fixed $\alpha > 0$ transforms an instance $(\cF',U')$ of 
\pallsparsesetcoverargs{$k$}{$c$}
into polynomially many instances of the $\psetcoverarg{$k'$}$ problem,
for some positive integer $k'$,
where the size $t$ of the solution in the resulting 
$\psetcoverarg{$k'$}$ instances
is at most $\alpha |U'|$.

In order to find the parity of the number of all set covers of the instance $(\cF',U')$
we find the parity of the number of set covers of a particular size.
That is we iterate over all possible sizes $j=1,\ldots,|\cF'|$ of a set cover.
Let us assume that we want to find the parity of the number of set covers of size $j$
and for each positive integer $j' < j$ we know the parity of the number of set covers of $(\cF',U')$ of size $j'$.
Let $q$ be the smallest power of two
satisfying $\frac{|\cF'|}{q}+2 \le \alpha|U'|$.
We assume that $\alpha|U'| \ge 3$ since otherwise the instance is small and we can solve it by brute force (recall that $\alpha$ is
a given constant).
Observe that $q$ is upper bounded by a constant independent of $|U'|$ since $|\cF'| \le c|U'|$.

We create a temporary set system $(\cF_0,U_0)$ to ensure that the size of the set covers we are looking for is divisible by $q$.
Let $r = j \bmod q$.
We make $(\cF_0,U_0)$ by taking the set system $(\cF',U')$
and adding $q-r$ new elements to the universe $U_0$ and also $q-r$ singleton sets of the new elements to the family $\cF_0$.
Now we are looking for the parity of the number of set covers of size $j_0=j+(q-r)$ in $(\cF_0,U_0)$.
Observe that for each $j' < j_0$ we know the parity of the number of set covers of size $j'$ in $(\cF_0,U_0)$ since it is equal to the parity of
set covers of $(\cF',U')$ of size $j'-(q-r) < j$
which we already know.

To obtain a \pallsetcoverarg{$k'$} instance we set $U^* = U_0$
and we form a family $\cF^*$ of all unions of
exactly $q$ sets from $\cF_0$, that is for each of $\binom{|\mathcal{F}_0|}{q}$ choices of $q$
sets $S_1,\ldots,S_q \in \cF_0$ we add to $\cF^*$ the set $\bigcup_{i=1}^q S_i$ (note that $\cF^*$ might be a multiset).
Finally we set $t^*=j_0/q$ which is an integer since $j+(q-r)$ is divisible by $q$.
Observe that $t^* \le \frac{j}{q}+1 \le \alpha|U'|-1$,
by the definition of $q$, but $(\cF^*,U^*,t^*)$ might not be a proper instance of \pallsetcoverarg{$kq$},
since $\cF^*$ could be a multiset.
Note that each subset of $U^*$ appears in $\cF^*$ at most $(2^{kq})^{q}=2^{kq^2}$
times, since $\cF_0$ has no duplicates and each set in $\cF^*$ is a union of exactly $q$ sets from $\cF_0$.
To overcome this technical obstacle
we make a new instance $(\cF,U,t)$,
where as $U$ we take $U^*$ with $z=1+kq^2$ elements
added, $U=U^* \cup \{e_1,\ldots,e_z\}$.
We use elements $\{e_1,\ldots,e_{z-1}\}$ to make sets from $\cF^*$ different
in $\cF$ by taking a different subset of $\{e_1,\ldots,e_{z-1}\}$ for duplicates.
Additionally we add one set $\{e_1,\ldots,e_z\}$ to the family $\cF$
and set $t=t^*+1$.
In this way we obtain $(\cF,U,t)$, that is a proper
\pallsetcoverarg{$(kq+z)$} instance and $t = t^*+1 \le \alpha|U'|$.
Observe that in the final instance we have $|U| \le n+q+z$
and $|\cF| \le (cn+q)^q+1$, which is a polynomial since $k,c,q$ and $z$ are
constants.

To summarize the reduction, we have taken an instance of 
\pallsparsesetcoverargs{$k$}{$c$}
and iterated over the size of solution.
Next we made the size divisible by $q$ by adding additional elements to the universe
and created a multiset family $\mathcal{F}^*$ from which we made
a set family by differentiating identical sets with additional elements
of the universe.
Our goal was to decide whether the \pallsetcoverarg{$k$} instance $(\cF',U')$ (for $k'=kq+z$)
has an odd number of set covers,
which means that we want to control the correspondence between the parity
of the number of solutions in each part of the construction.
Observe that the only step of the construction which has
nontrivial correspondence between the number of solutions
of the former and the latter instance is the grouping step
where we transform an instance $(\cF_0,U_0,j_0)$ into a multiset
instance $(\cF^*,U^*,t^*)$.

Hence we assume that we know the parity of the number of set covers
of size $t^*=j_0/q$ in $(\cF^*,U^*)$
as well as the parity of the number of set covers
of size $j'$ for each $j' < j_0$ in $(\cF_0,U_0)$.
Our objective is to compute the parity of the number of set covers
of size $j_0$ in $(\cF_0,U_0)$ in polynomial time
and for this reason we introduce a few definitions and lemmas.
Recall that each set in $\cF^*$ corresponds to a union of exactly $q$ sets in $\cF_0$
and let $\Gamma\colon \cF^* \rightarrow 2^{\cF_0}$ be a function
that for each set in $\cF^*$ assigns a family of exactly $q$ sets
from $\cF_0$ that it was made of.
Moreover let $\cS^* \subseteq 2^{\cF^*}$ be the family of set covers of size $t^*$ in $(\cF^*,U^*)$
and let $\cS_0 \subseteq 2^{\cF_0}$ be the set of set covers of size {\emph{at most}}
$j_0$ in $(\cF_0,U_0)$.
We construct a mapping $\Phi\colon \cS^*\rightarrow \cS_0$ which maps each
set cover $A \in \cS^*$ to a set cover $A_0 \in \cS_0$
such that $A_0$ is exactly the set of sets from $\cF_0$
used in the $t^*$ unions of $q$ sets from $\cF_0$,
that is $\Phi(A)=\bigcup_{X \in A} \Gamma(X)$.
In the following lemma we prove that for a set cover $A_0 \in \cS_0$
the size of $\Phi^{-1}(A_0)$ depends solely on the size of $A_0$.

\begin{lemma}
\label{lem:sizedependent}
Let $A_0,B_0 \in \cS_0$ such that $|A_0| = |B_0|$.
Then $|\Phi^{-1}(A_0)|=|\Phi^{-1}(B_0)|$.
\end{lemma}

\begin{proof}
Let $A_0 = \{X_1,\ldots,X_a\}$ be a set from $\cS_0$,
where each $X_i \in \cF_0$.
Observe that for any $A \in \cS^*$ we
have $\Phi(A) = A_0$ if and only if $\bigcup_{i=1}^a \Gamma(X_i) = A$.
Consequently $|\Phi^{-1}(A_0)|$ is equal to the number of set covers of size $t^*$ in the set system $(\binom{A_0}{q},A_0)$
and hence $|\Phi^{-1}(A_0)|$ depends only on the size of $A_0$.
\end{proof}

Now we prove that for each set cover $A_0 \in \cS_0$ of size
$j_0$ an odd number of set covers from $\cS^*$ is mapped by $\Phi$ to $A_0$.

\begin{lemma}
\label{lem:binom}
For any nonnegative integers $a,b$ such that $b \le a$
the binomial coefficient $\binom{a}{b}$ is odd if and only if $\ones(b) \subseteq \ones(a)$,
where $\ones(x)$ is the set of indices containing
ones in the binary representation of $x$.
\end{lemma}

\begin{proof}
For a nonnegative integer $x$ by $f(x)$ let
us denote the greatest integer $i$ such that $x!$ is divisible by $2^i$,
that is
\begin{align*}
 f(x) &=  \sum_{i \ge 1} \big\lfloor \frac{x}{2^i} \big\rfloor \\
      &= (\sum_{i \ge 1} \frac{x}{2^i}) - \frac{1}{2} \cdot |\{ i \ge 1 : \big\lfloor\frac{x}{2^{i-1}}\big\rfloor \textrm{ is odd}\}| \\
      &= (\sum_{i \ge 1} \frac{x}{2^i}) - \frac{|\ones(x)|}{2}
\end{align*}

Since $\binom{a}{b}=\frac{a!}{b!(a-b)!}$ we infer
that $\binom{a}{b}$ is odd if and only if $f(a)=f(b)+f(a-b)$,
which by the above formula is equivalent to $|\ones(a)| = |\ones(b)| + |\ones(a-b)|$.
However for any nonnegative integers $x,y$ we have $\ones(x+y) \le \ones(x)+\ones(y)$
and moreover $\ones(x+y) = \ones(x)+\ones(y)$ if and only if there are no carry-operations
when adding $x$ to $y$, which is equivalent to $\ones(x) \cap \ones(y) = \emptyset$.

Therefore by setting $x=b$ and $y=a-b$ we infer
that $\binom{a}{b}$ is odd if and only if $\ones(b) \cap \ones(a-b) = \emptyset$
which is equivalent to $\ones(b) \subseteq \ones(a)$ and the lemma follows.
\end{proof}

\begin{lemma}
\label{lem:oddpartition}
Let $A_0 \in \cS_0$ such that $|A_0|=j_0$
then $|\Phi^{-1}(A_0)|$ is odd.
\end{lemma}

\begin{proof}
Since $|\Phi^{-1}(A_0)|$ is equal to the number of set covers of size $t^*$ in the set system $(\binom{A_0}{q},A_0)$
and $|A_0|=j_0=t^*q$ we infer that $|\Phi^{-1}(A_0)|$ is equal to the number of unordered partitions of $A_0$ into sets of size $q$.
Hence $|\Phi^{-1}(A_0)|=\prod_{i=0}^{t^*-1}\binom{j_0-1-iq}{q-1}$.
Since $j_0$ is divisible by $q$ and $q$ is a power of two
using Lemma~\ref{lem:binom} we have $|\Phi^{-1}(A_0)| \equiv 1\ (\bmod\ 2)$.
\end{proof}

For $j=1,\ldots,j_0$ by $s_j$ let us denote the parity of the number of set covers of $(\cF_0,U_0)$ of size $j$ modulo $2$.
Recall that we know the value of $s_j$ for each $j < j_0$ and we want to compute $s_{j_0}$
knowing also $|\cS^*| \bmod 2$.
By Lemma~\ref{lem:sizedependent} we can define $d_j$ for $j=1,\ldots,j_0$,
that is the value of $|\Phi^{-1}(A_0)| \bmod 2$ for a set $A_0 \in \cS_0$ of size $j$.
By Lemma~\ref{lem:oddpartition} we know that $d_{j_0}$ equals one.
Thus we have the following congruence modulo $2$.
\begin{eqnarray*}
|\cS^*| & = & \sum_{A_0 \in \cS_0} |\Phi^{-1}(A_0)| \equiv
\sum_{j=1}^{j_0} s_jd_j = s_{j_0} + \sum_{j=1}^{j_0-1} s_jd_j\,.
\end{eqnarray*}
Hence knowing $|\cS^*| \bmod 2$ and all values $s_j$ for $j < j_0$ in
order to compute $s_{j_0}$ it is enough to compute all the values $d_j$,
what we can do in polynomial time thanks to the following lemma.

\begin{lemma}
For each $j=1,\ldots,j_0$ we can calculate the value of $d_j$
in polynomial time.
\end{lemma}

\begin{proof}
Again we use that fact that for a set $A_0 \in \cS_0$
we have that $|\Phi^{-1}(A_0)|$ is equal to the number set covers of size $t^*$ in the set system $(\binom{A_0}{q},A_0)$.
Using the inclusion-exclusion principle modulo two we obtain the following formula when $|A_0|=j$.
\[
	|\Phi^{-1}(A_0)| \equiv \sum_{X \subseteq A_0} \left|\left\{\mathcal{H} \subseteq \binom{X}{q} \Big{|} |\mathcal{H}|=t^*\right\}\right|= \sum_{i=0}^{j} \binom{j}{i} \binom{\binom{i}{q}}{t^*},
\]
Where the second equality follows by grouping all summands $X \subseteq A_0$ with $|X|=i$ for every $0 \leq i \leq |A_0|$.
\end{proof}

Consequently, by solving a polynomial of $n$ number of instances
of the \psetcoverarg{$k'$} problem with universe size bounded by 
$n+q+z$ and set family size bounded by $(cn+q)^q+1$, we verify whether 
the initial set system $\mathcal{F}' \subseteq 2^{U'}$ has an odd 
number of set covers, which finishes the proof of 
Theorem~\ref{thm:psetcovera}.
\qed

\section{Summary and Open Problems}
We have shown that the exponential time complexity of a number of
basic problems is strongly interconnected. Specifically, our results
imply that the optimal growth rates of a a number of problems are in
fact asymptotically equal.
Assuming \seth{}, our results imply tight lower bounds on the growth
rates for a number of search problems whose growth rates are achieved
by na\"ive brute force algorithms.
For problems solvable by dynamic programming, we gave tight lower
bounds assuming that the optimal growth rate of~\setcover is achieved
by its known dynamic programming algorithm.
Finally, we connected the two types of results by showing that
\seth{} implies tight lower bounds on the optimal growth rates of
corresponding parity variants.
We conclude our work with some open problems.
\begin{enumerate}\setlength\itemsep{-2mm}
\item Is it possible to rule out an algorithm for \setcover with running time $2^{\epsilon n}m^{O(1)}$, $\epsilon < 1$, assuming $\seth$?
\item Is it possible to rule out an algorithm for \coloring{} with running time $2^{\epsilon n}$, $\epsilon < 1$, assuming $\seth$? What about a lower bound for \coloring under the assumption that there does not exist a $\delta < 1$ such that \setcover with sets of size at most $k$ has a $O(2^{\delta n}m^{O(1)})$ time algorithm for every $k$?
\item Is it possible to rule out an algorithm that {\em counts} the number of proper $c$-colorings of an input graph in time $2^{\epsilon n}$, $\epsilon < 1$ assuming $\pseth$?
\item Assuming $\seth$, is it possible to rule out an algorithm with running time $2^{\epsilon n}n^{O(1)}$, $\epsilon < 1$ for the satisfiability of circuits with at most $cn$ gates of {\em unbounded fan in}, for a concrete constant c?
\item Assuming $\seth$, is it possible to rule out an algorithm with 
  running time $O(c^n)$ for \cnfsatargs{3} for a concrete constant 
  $c$?
\end{enumerate}

\newpage

\printbibliography
%

\end{document}